\documentclass[10pt]{article}

\usepackage{graphicx}
\usepackage{amsmath,amssymb,amsthm}
\usepackage[ruled, vlined]{algorithm2e}
\usepackage{rotating}
\usepackage{hyperref}

\usepackage{booktabs}
\usepackage{mathtools}

\theoremstyle{plain}
\newtheorem{theorem}{Theorem}
\newtheorem{lemma}{Lemma}

\newtheorem{proposition}{Proposition}

\theoremstyle{definition}

\theoremstyle{remark}

\newcommand{\reals}{\mathbb{R}}

\newcommand{\vect}[1]{\mbox{\boldmath$#1$}}

\title{Incremental Network Design with Maximum Flows}

\author{
Thomas Kalinowski$^\dagger$\and Dmytro Matsypura$^\star$\and Martin W.P. Savelsbergh$^\dagger$\\[.5ex]
$^\dagger${\it University of Newcastle, Australia}\qquad $^\star${\it University of Sydney, Australia}
}

\begin{document}

\maketitle

\begin{abstract}
We study an incremental network design problem, where in each time
period of the planning horizon an arc can be added to the network
and a maximum flow problem is solved, and where the objective is to
maximize the cumulative flow over the entire planning horizon. After
presenting two mixed integer programming (MIP) formulations for this
NP-complete problem, we describe several heuristics and prove
performance bounds for some special cases. In a series of
computational experiments, we compare the performance of the MIP
formulations as well as the heuristics.
\end{abstract}

\section{Introduction}\label{sec:intro}

In the planning process for many network infrastructures, when the
network is constructed over a significant period of time, the
properties of the intermediate partial networks have to be taken
into account. \emph{Incremental network design}, introduced in
\cite{baxter2014incremental}, represents a class of optimization
problems capturing that feature and combining two types of
decisions: which arcs should be added to a given network in order to
achieve a certain goal, and when should these arcs be added?

Variants of this problem have been studied in diverse contexts, for
instance, the design of transportation
networks~\cite{kim2008sequencing,UkPa2009}, network infrastructure
restoration after disruptions due to natural disasters
~\cite{cavdaroglu2013integrating,lee2007restoration,lee2009network},
and the transformation of an electrical power grid into a smart
grid~\cite{mahmood2008design,momoh2009smart}.

A general class of mathematical optimization problems that captures
essential features of the described decision problems, and includes
the problem discussed in this paper as a special case, was
introduced in~\cite{nurre2012restoring,nurre2013integrated}, where
an \emph{integrated network design and scheduling} problem is
specified by (1) a scheduling environment that describes the
available resources for adding arcs to the network, (2) a
performance measure, which prescribes how a given network is
evaluated (for instance by the shortest $s$-$t$ path or by the
maximum $s$-$t$ flow in the network), and (3) by the optimization
goal, which is either to reach a certain level of performance as
quickly as possible or to optimize the cumulative performance over
the entire planning horizon.

We focus on the special case corresponding to incremental network
design as introduced in~\cite{baxter2014incremental}. Our
scheduling environment is such that at most one arc can be added to
the network in each time period, and we optimize the cumulative
performance, i.e., the sum of the performance measures of the
networks in all time periods. Even in this simple setting, the
problem has been shown to be NP-complete for classical network
optimization problems: for the shortest $s$-$t$ path problem
in~\cite{baxter2014incremental}, and for the maximum $s$-$t$ flow
problem in~\cite{nurre2013integrated}. Interestingly, the
incremental variant of the minimum spanning tree problem can be
solved efficiently by a greedy
algorithm~\cite{engel2013incremental}, while it becomes NP-complete
in the more general setup of~\cite{nurre2013integrated}. The
performance measure considered in this paper is the value of a
maximum $s$-$t$ flow.

In Section~\ref{sec:problem}, we introduce notation, state the
problem precisely, and present two MIP formulations. In
Section~\ref{sec:heuristics}, we describe three heuristics, the
first one seeks to increment the flow as quickly as possible, the
second seeks to reach a maximum flow as quickly as possible, and the
third one is a hybrid of the first two. In
Section~\ref{sec:unit_capacity}, we prove performance guarantees for
the first two heuristics in special cases: for unit capacity
networks they provide a 2-approximation algorithm and a
$3/2$-approximation algorithm, respectively, and these bounds can be
strengthened when the network has a special structure.
Section~\ref{sec:computations} discusses the results of a set of
computational experiments using randomly generated instances. After
describing the instance generation, we compare the performance of
the two MIP formulations, and evaluate and compare the heuristics on
hard instances. We end, in Section~\ref{sec:final}, with some final
remarks.

\section{Problem formulation}\label{sec:problem}

We are given a network $D=(N,A)$ with node set $N$ and arc set
$A=A_e\cup A_p$, where $A_e$ contains existing arcs and $A_p$
contains potential arcs, as well as a source $s\in N$ and sink $t\in
N$. For each arc, we are given an integer capacity $u_a>0$, and for every
node $v\in N$, we denote with $\delta^{\text{in}}(v)$ and
$\delta^{\text{out}}(v)$ the set of arcs entering $v$ and the set of
arcs leaving $v$, respectively. Let $T>\lvert A_p\rvert$  be the
length of the planning horizon. In every period, we have the option
to expand the useable network, which initially consists of only the
existing arcs, by ``building'' a single potential arc $a\in A_p$,
which will be available for use in the following time period. In
every period, the value of a maximum $s$-$t$ flow is recorded (using
only useable arcs, i.e., existing arcs and potential arcs that have
been built in previous periods). The objective is to maximize the
total flow over the planning horizon. Note that the length of the
planning horizon ensures that every potential arc can be built. We
refer to a maximum $s$-$t$ flow using only existing arcs as an
\emph{initial maximum flow}, and to a maximum flow for the complete
network as an \emph{ultimate maximum flow}. This problem is strongly
NP-hard~\cite{nurre2013integrated} even when restricted to instances
where every existing arc has capacity $1$ and every potential arc
has capacity $3$. (A simple proof of this result can be found in the
appendix.)

The problem can be formulated as a mixed integer program. For every
$a\in A$ and $k \in \{1,\ldots,T\}$, we have a flow variable $x_{ak}
\geqslant 0$, and for every $a\in A_p$ and $k \in \{1,\ldots,T\}$,
we have a binary variable $y_{ak}$ which indicates if arc $a$ is
built before period $k$ ($y_{ak} = 1$) or not ($y_{ak} = 0$). The
incremental maximum flow problem is then
\begin{equation*}
\max\ \sum_{k=1}^T\left(\sum_{a\in
\delta^{\text{out}}(s)}x_{ak}-\sum_{a\in
\delta^{\text{in}}(s)}x_{ak}\right)
\end{equation*}
subject to
\begin{align*}
&& \sum_{a\in\delta^{\text{out}}(v)}x_{ak}-\sum_{a\in\delta^{\text{in}}(v)} x_{ak} &=0 && \text{for }v\in N\setminus\{s,t\},\ k\in\{1,\ldots,T\}, \\
&& x_{ak} &\leqslant u_a && \text{for }a\in A_e,\ k\in\{1,\ldots,T\},\\
&& x_{ak} &\leqslant u_ay_{ak} && \text{for }a\in A_p,\ k\in\{1,\ldots,T\},\\
&& y_{ak} &\geqslant y_{a,k-1} && \text{for }a\in A_p,\ k\in\{2,\ldots,T\},\\
&& y_{a1} &= 0 && \text{for }a\in A_p,\\
&& \sum_{a\in A_p}(y_{ak}-y_{a,k-1}) &\leqslant 1 && \text{for }k\in\{2,\ldots,T-1\},\\
&& x_{ak} &\geqslant 0 && \text{for }a\in A,\ k\in\{1,\ldots,T\},\\
&& y_{ak} &\in\{0,1\} && \text{for }a\in A_p,\ k\in\{1,\ldots,T\}.
\end{align*}
We denote this formulation by IMFP$^1$.

A potential weakness of IMFP$^1$ is that it may suffer from
symmetry. If multiple arcs need to be build to increase the maximum
$s$-$t$ flow, then the order in which these arcs are build does not
matter, which introduces alternative, symmetrical solutions. Next,
we present an alternative MIP formulation which avoids this
difficulty. Let $f$ and $F$ denote the initial and the ultimate
maximum flow value, respectively, and let $r = F - f$. We introduce
binary variables $y_{ak}$ for $a \in A_p$ and $k = 1,2, \ldots, r$
with the interpretation
\[y_{ak}=\begin{cases}
1 & \text{if arc $a$ is build while the max flow value is less than $f+k$},\\
0 & \text{otherwise}.
\end{cases}\]
The number of time periods with maximum flow value $f$ is
$\sum_{a\in A_p}y_{a1}$, and for $k=1,\ldots,r-1$, the number of
time periods with maximum flow value $f+k$ is $\sum_{a\in
A_p}(y_{a,k+1}-y_{ak})$. Consequently, the total flow is
\begin{multline*}
f \sum_{a\in A_p} y_{a1} + \sum_{k=1}^{r-1} (f+k)\sum_{a\in A_p}
(y_{a,k+1}-y_{ak}) + F \left(T-\sum_{a\in A_p}y_{ar}\right) \\
=TF+\sum_{a\in A_p}\sum_{k=1}^ry_{ak}\left[(f+k-1)-(f+k)\right]
= TF -
\sum_{a\in A_p} \sum_{k=1}^r y_{ak}.
\end{multline*}
Hence the incremental maximum flow problem can also be formulated as
follows
\begin{equation*}
\min\ \sum_{a\in A_p}\sum_{k=1}^ry_{ak}
\end{equation*}
subject to
\begin{align*}
&& \sum_{a\in\delta^{\text{out}}(v)}x_{ak}-\sum_{a\in\delta^{\text{in}}(v)} x_{ak} &=
\begin{cases}
0 & \text{for }v\not\in\{s,t\} \\
f+k & \text{for }v=s\\
-f-k & \text{for }v=t
\end{cases} && \text{for }v\in N,\ k\in\{1,\ldots,r\}, \\
&& x_{ak} &\leqslant u_a && \text{for }a\in A_e,\ k\in\{1,\ldots,r\},\\
&& x_{ak} &\leqslant u_ay_{ak} && \text{for }a\in A_p,\ k\in\{1,\ldots,r\},\\
&& y_{ak} &\leqslant y_{a,k+1} && \text{for }a\in A_p,\ k\in\{1,\ldots,r-1\},\\
&& x_{ak} &\geqslant 0 && \text{for }a\in A,\ k\in\{1,\ldots,r\},\\
&& y_{ak} &\in\{0,1\} && \text{for }a\in A_p,\ k\in\{1,\ldots,r\}.
\end{align*}
We denote this formulation by IMFP$^2$.

Observe that the size of IMFP$^1$ strongly depends on the length of
the planning horizon, whereas the size of IMFP$^2$ strongly depends
on the difference between the initial and ultimate maximum flow
values.

\section{Heuristics}\label{sec:heuristics}

In this section, we introduce two natural strategies for trying to
obtain high quality solutions: (1) getting a flow increment as
quickly as possible, and (2) reaching a maximum possible flow as
quickly as possible, as well as a hybrid strategy.

\subsection{Quickest flow increment}\label{subsec:quickest_inc}

A natural greedy strategy is to build the arcs such that a flow
increment is always reached as quickly as possible. Suppose we have
already built the arcs in $B \subseteq A_p$ to reach a maximum flow
value $f+k$. A smallest set of potential arcs to be built, in
addition to $B$, to reach a flow of value at least $f+k+1$ can be
determined by solving a fixed charge network flow problem:
find a flow of value $f+k+1$ where arcs in $A_e \cup B$ have zero
cost, and arcs in $A_p \setminus B$ incur a cost of 1 if they carry
a nonzero flow. More formally, in order to determine the smallest
number of potential arcs that have to be built to increase the flow
from $f+k$ to at least $f+k+1$, we solve the problem
\textsc{MinArcs}($B,k+1$):
\begin{align*}
&&\min\ z=\sum_{a\in A_p\setminus B}y_{a} \\
&\text{subject to} & \sum_{a\in\delta^{\text{out}}(v)}x_{a}-\sum_{a\in\delta^{\text{in}}(v)} x_{a} &=
\begin{cases}
0 & \text{for }v\not\in\{s,t\} \\
f+k+1 & \text{for }v=s\\
-f-k-1 & \text{for }v=t
\end{cases} && \text{for }v\in N, \\
&& x_{a} &\leqslant u_a && \text{for }a\in A_e\cup B,\\
&& x_{a} &\leqslant u_ay_{a} && \text{for }a\in A_p\setminus B,\\
&& x_{a} &\geqslant 0 && \text{for }a\in A,\\
&& y_{a} &\in\{0,1\} && \text{for }a\in A_p\setminus B.
\end{align*}
Next, among all the possibilities that increase the flow by building
the smallest possible number of potential arcs, we want to choose
one that maximizes the flow increase. Given the optimal value $z^*$
for the problem \textsc{MinArcs}($B,k+1$), this can be achieved by
solving another MIP, denoted by \textsc{MaxVal}($B,z^*$):
\begin{align*}
&&\max\ \xi\phantom{\sum_{a\in\delta^{\text{out}}(v)}x_{a}-\sum_{a\in\delta^{\text{in}}(v)} x_{a}} &  \\
&\text{subject to} & \sum_{a\in\delta^{\text{out}}(v)}x_{a}-\sum_{a\in\delta^{\text{in}}(v)} x_{a} &=
\begin{cases}
0 & \text{for }v\not\in\{s,t\} \\
\xi & \text{for }v=s\\
-\xi & \text{for }v=t
\end{cases} && \text{for }v\in N, \\
&& x_{a} &\leqslant u_a && \text{for }a\in A_e\cup B,,\\
&& x_{a} &\leqslant u_ay_{a} && \text{for }a\in A_p\setminus B,\\
&& \sum_{a\in A_p\setminus B} y_a &= z^*\\
&& x_{a} &\geqslant 0 && \text{for }a\in A,\\
&& y_{a} &\in\{0,1\} && \text{for }a\in A_p\setminus B.
\end{align*}

The greedy heuristic is described formally in
Algorithm~\ref{alg:greedy} (assuming the entire set $A_p$ is
provided as input).
\begin{algorithm}[htb]
\SetKwInOut{Input}{input}\SetKwInOut{Output}{output}
\Input{$\overline{A}_p \subset A_p$}
$D \leftarrow (V,A)$ with $A = A_e \cup \overline{A}_p$ \;
$B \leftarrow \emptyset$ \{ \emph{initialize the set of built arcs} \} \;
$k \leftarrow 0$ \{ \emph{index of initial flow value} \} \;
\While{$k < F-f$}{
    $z^* \leftarrow$ optimal value of \textsc{MinArcs}($B,k+1$) \;
    $(x,y)\leftarrow$ optimal solution for \textsc{MaxVal}($B,z^*$) \;
    \{ \emph{add potential arcs to be built to reach a flow of at least $f+k+1$} \} \;
    $B \leftarrow B \cup \{a \in \overline{A}_p \ : \ y_a = 1 \}$ \;
    $k \leftarrow$ (maximum flow value using only arcs in $A_e \cup B$) $-f$ \;
}
\caption{Quickest-increment\label{alg:greedy}}
\end{algorithm}
In general, this heuristic still requires the solution of MIPs
\textsc{MinArcs}($B,k+1$) and \textsc{MaxVal}($B,z^*$), hence we do
not get a polynomial bound on the runtime. But these MIPs are much
smaller than IMFP$^1$ and IMFP$^2$, so we might expect that they can
be solved more efficiently.

However, if the capacities of all potential arcs are equal to $1$,
then the problems \textsc{MinArcs}($B,k+1$) reduce to minimum cost
flow problems, and their solutions are already optimal for
\textsc{MaxVal}($B,z^*$), because the optimal flow increment is 1.
Thus for these instances, the quickest increment heuristic runs in
polynomial time.

Furthermore, if the maximization of the flow increment is omitted,
then the heuristic can be implemented to run in polynomial time even
for general capacities. For this purpose, suppose we have already
built the arcs in $B\subset A_p$ and we have a maximum flow $\vect
x$ of value $f+k$ for the network $(V,A_e\cup B)$. Using the
residual network with respect to $\vect x$, we can apply a labeling
procedure which computes for each node $v$ a triple
$(d(v),\delta(v),p(v))$, where
\begin{itemize}
\item $d(v)$ is the minimum number of arcs in $A_p\setminus B$ on any augmenting $s$-$v$-path,
\item $\delta(v)$ is the maximum augmentation for any $s$-$v$-path using at most $d(v)$ arcs from
  $A_p\setminus B$, and
\item $p(v)$ is the predecessor of $v$ on an $s$-$v$-path which contains $d(v)$ arcs from
  $A_p\setminus B$ and can be augmented by $\delta(v)$ units of flow.
\end{itemize}
This can be done in time $O(\lvert A\rvert)$, for instance using BFS
starting from $s$, and the distance label $d(t)$ equals the minimum
size of a set $B'\subseteq A_p\setminus B$ with the property that
the maximum $s$-$t$-flow value in the network $(V,A_e\cup B\cup B')$
is strictly larger than $f+k$. Furthermore, $B'=P\cap(A_p\setminus
B)$ is such a set of size $d(t)$, where $P$ is the augmenting
$s$-$t$-path corresponding to the result of the labeling procedure.
Suppose there is a set $B''\subseteq A_p\setminus B$ with $\lvert
B''\rvert<d(t)$ such that the maximum flow in $(V,A_e\cup B\cup
B'')$ is strictly larger than $f+k$. Then the residual network for
$(V,A_e\cup B\cup B'')$ with respect to the flow $\vect x$ of value
$F_B$ must contain an augmenting path $P'$. Now $P'\cap(A_P\setminus
B)\subseteq B''$, thus $P'$ is an augmenting $s$-$t$-path which
contains less than $d(t)$ arcs from $A_p\setminus B$, which
contradicts the construction of $d(t)$.  Using this labeling
procedure, a quickest-increment heuristic can be described as
follows.
\begin{itemize}
\item Initialize the set of built arcs $B=\emptyset$ and a zero flow $\vect x=\vect 0$
\item while $\vect x$ is not a maximum flow for $(V,A)$ do
  \begin{itemize}
  \item Find an augmenting path $P$ using the smallest number of new potential arcs.
  \item Add the potential arcs on $P$ to $B$.
  \item Update the flow $\vect x$.
  \end{itemize}
\end{itemize}
This algorithm runs in time $O(\lvert V\rvert\lvert A\rvert^3)$. The
$\delta$-component in our node labels ensures that we find the
augmenting path that (1) requires the smallest number of new arcs to
be built, and (2) allows the maximum augmentation among these paths.
Nevertheless, it might be beneficial to build another path which
allows more augmentations afterwards without building additional
arcs. This is illustrated in Figure~\ref{fig:example}.
\begin{figure}[htb]
  \centering
\includegraphics[width=.4\linewidth]{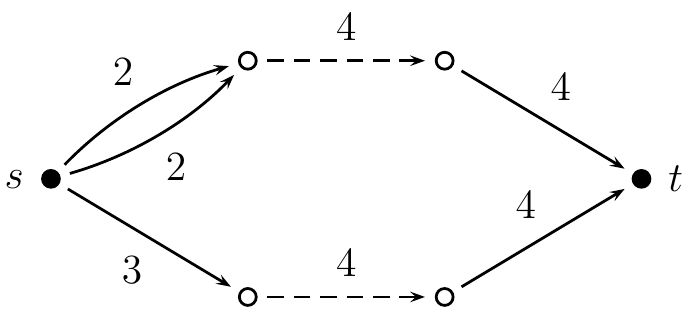}
  \caption{An example where the polynomial variant of \textit{Quickest-increment} builds the
    ``wrong'' arc. Every augmenting path requires one
    potential arc. Our algorithm chooses the lower arc because it allows the augmentation
    of 3 units of flow. But it would be better to build the upper arc first, and
    Algorithm~\ref{alg:greedy} does the right thing.}\label{fig:example}
\end{figure}

Since Algorithm~\ref{alg:greedy}, which solves subproblems
\textsc{MinArcs}$(B,k+1)$ and \textsc{MaxVal}$(B,z^*)$, when
implemented using state-of-the-art commercial MIP solver, turned out
to be sufficiently efficient for our test instances, we decided not
to implement the polynomial variant for our computational study,
which is presented in Section~\ref{sec:computations}.

\subsection{Quickest ultimate flow}\label{subsec:quickest_ultimate}

A potential drawback of the quickest increment heuristic is that it
might build arcs which yield (small) flow increments but that are
not used in an ultimate maximum flow. To force the heuristic to
build arcs that are used in an ultimate maximum flow, we can first
determine a smallest set of potential arcs that admit an ultimate
maximum flow by finding an optimal solution $(x^*,y^*)$ to the
problem \textsc{MinArcs}($\emptyset,r$), and then restricting the
arc choices in Algorithm~\ref{alg:greedy} to the arcs used in the
optimal solution $(x^*,y^*)$, i.e., by fixing $y_a = 0$ for all $a
\in A_p$ with $y^*_a = 0$. This algorithm is described formally in
Algorithm~\ref{alg:quickest_ultimate}.
\begin{algorithm}[htb]
\SetKwInOut{Input}{input}\SetKwInOut{Output}{output}
\Input{$r$}
$(x^*,y^*)\leftarrow$ optimal solution for
\textsc{MinArcs}$(\emptyset,r)$ \; $\overline{A}_p \leftarrow \{ a
: y^*_a = 1 \} $ \; $(\tilde x,\tilde y) \leftarrow$ solution
obtained by \textit{Quickest-increment} with input $\overline{A}_p$
\; \caption{Quickest-to-ultimate\label{alg:quickest_ultimate}}
\end{algorithm}
Because the smallest set of potential arcs that admits an ultimate
maximum flow is unlikely to be unique, it is possible that the
initial choice of arcs may not be the best.

\subsection{Quickest target flow}

Combining the ideas of quickest flow increment and quickest ultimate
flow, we can design a general framework. Let $0 = r_0 < r_1 < r_2 <
\cdots < r_k = r$ be a sequence of target values. Put $B_0 =
\emptyset$, and for $i \geqslant 1$ suppose that we have already
determined a set $B_{i-1} \subseteq A_p$ such that the maximum flow
value for $A_e\cup B_{i-1}$ is $f + r_{i-1}$. Following the quickest
ultimate flow strategy, we can determine a smallest set of potential
arcs $B_i$ with $B_{i-1} \subseteq B_i \subseteq A_p$ such that $A_e
\cup B_i$ allows a flow of value $f + r_{i}$, and then we can use
the quickest flow increment strategy to specify the order in which
the arcs in $B_i\setminus B_{i-1}$ are built. This is described in
detail in Algorithm~\ref{alg:quickest_target}.
\begin{algorithm}[htb]
\SetKwInOut{Input}{input}\SetKwInOut{Output}{output}
\Input{$0 = r_0 < r_1 < r_2 < \cdots < r_k = r$}
$B_0 \leftarrow \emptyset$ \;
$\overline{A}_p \leftarrow \emptyset$ \;
\For{$i = 1, \ldots, p$}{
    $z^* \leftarrow$ optimal objective value for \textsc{MinArcs}$(B_{i-1},r_i)$ \;
    $(x^*,y^*)\leftarrow$ optimal solution for \textsc{MaxVal}$(B_{i-1},z^*)$ \;
    $B_i \leftarrow \{ a : y^*_a = 1 \} $ \;
    $\overline{A}_p \leftarrow \overline{A}_p \cup B_i $ \;
    $(\tilde x,\tilde y) \leftarrow$ solution obtained by \textit{Quickest-increment} with input $\overline{A}_p$ \;
} \caption{Quickest-to-target\label{alg:quickest_target}}
\end{algorithm}
Note that this framework implicitly also provides a risk mitigation
strategy, because we are not locked in to a set of arcs for the
entire planning period, as in \textit{Quickest-to-ultimate}.

\section{The unit capacity case}\label{sec:unit_capacity}

As mentioned in the introduction, the incremental network design
problem with maximum flows is strongly NP-hard even when restricted
to instances where every existing arc has capacity 1 and every
potential arc has capacity 3. That leaves open the possibility that
when all arcs have capacity 1, i.e., both existing and potential
arcs, the problem is polynomially solvable. Unfortunately, we have
been unable to settle the complexity status of the unit-capacity
case, but we have been able to derive a number of results regarding
the performance of \textit{Quickest-increment} and
\textit{Quickest-to-ultimate} in the unit-capacity case, which are
presented in this section.

Figures~\ref{fig:example_1} and~\ref{fig:example_2} illustrate that
the heuristics can be off by a factor of close to 2
(\textit{Quickest-to-ultimate}) and 3/2
(\textit{Quickest-increment}).
\begin{figure}[htb]
  \begin{minipage}[b]{.48\linewidth}
  \centering
\includegraphics[width=\textwidth]{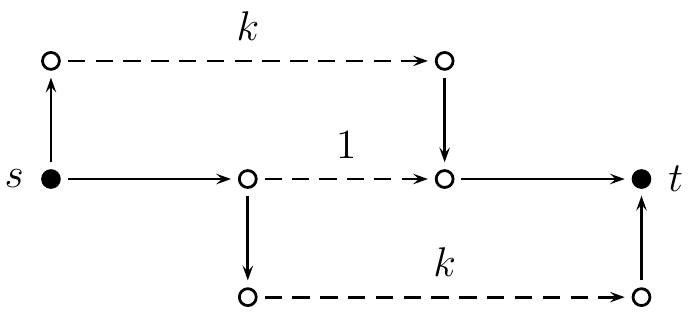}
\caption{Bad instance for \textit{Quickest-to-ultimate}.}\label{fig:example_1}
\end{minipage}\hfill
  \begin{minipage}[b]{.48\linewidth}
  \centering
\includegraphics[width=\textwidth]{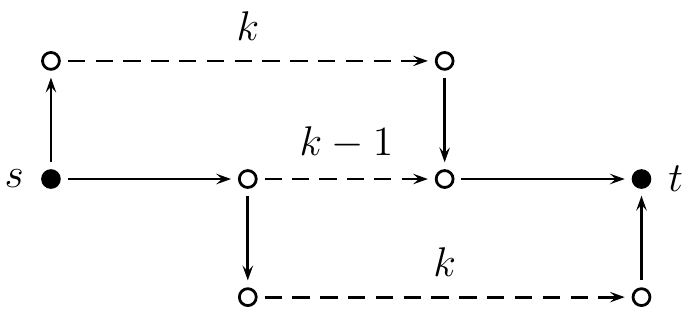}
  \caption{Bad instance for \textit{Quickest-increment}.}\label{fig:example_2}
  \end{minipage}
\end{figure}
The dashed arcs represent paths of potential arcs whose lengths are
indicated by the arc labels. For the instance in
Figure~\ref{fig:example_1}, the time horizon is $T=2k+2$ and the
optimal solution (which is found by \textit{Quickest-increment}) has
value $2k+2$, while \textit{Quickest-to-ultimate} provides a
solution of value $k\cdot 0+k\cdot 1+2\cdot 2= k+4$. For the
instance in Figure~\ref{fig:example_2}, the time horizon is $T=3k$
and the optimal solution (which is found by
\textit{Quickest-to-ultimate}) has value $3k$, while
\textit{Quickest-increment} provides a solution of value $(k-1)\cdot
0+2k\cdot 1+ 2 = 2k+2$. Combining the two examples we get an
instance that fools both heuristics (Figure \ref{fig:example_3}).
For this instance, \textit{Quickest-to-ultimate} achieves a total
flow of $10k+4$, \textit{Quickest-increment} gets $11k+3$, but a
combination of the two strategies yields a total flow of $13k$.
\begin{figure}[htb]
  \centering
\includegraphics[width=.7\textwidth]{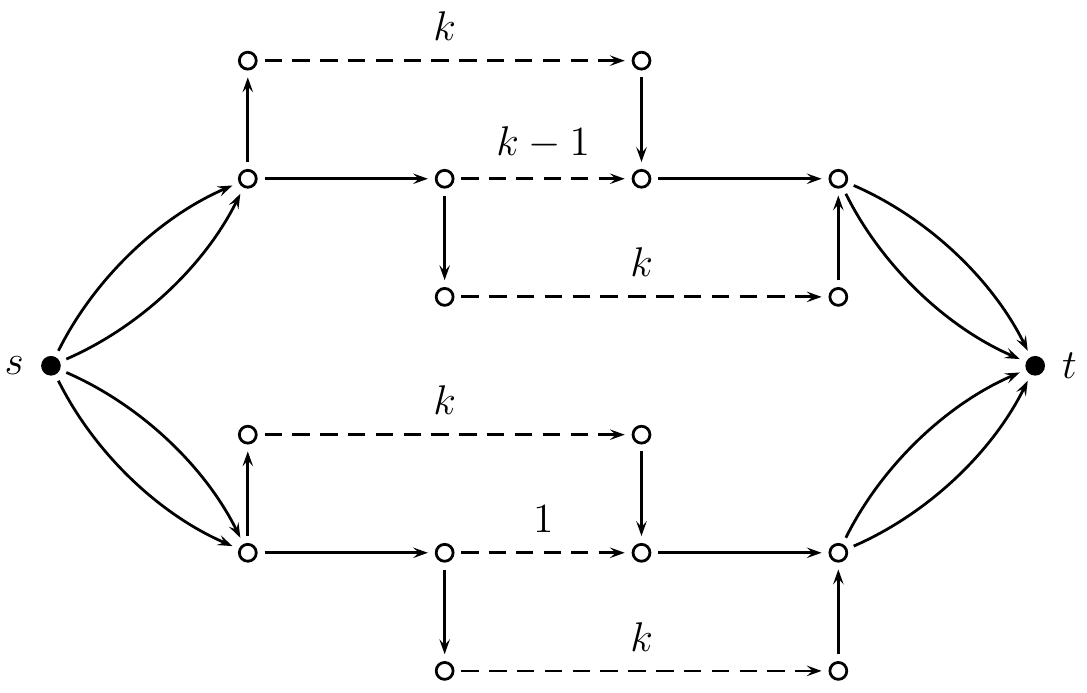}
\caption{An instance where both \textit{Quickest-to-ultimate} and
\textit{Quickest-increment} fail.}\label{fig:example_3}
\end{figure}

\subsection{Upper bounds}\label{subsec:upper_bounds}

In this subsection, we show that the instances in
Figures~\ref{fig:example_1} and~\ref{fig:example_2} represent the
worst situations for the heuristics \textit{Quickest-to-ultimate}
and \textit{Quickest-increment}, respectively. Let $z^*$, $z_1$ and
$z_2$ denote the optimal objective value, and the values obtained by
\textit{Quickest-to-ultimate} and \textit{Quickest-increment},
respectively. For $i=0,\ldots,r-1$, let $\lambda_i$ and $\mu_i$
denote the number of time periods with a maximum flow value $f+i$ in
the solution from \textit{Quickest-to-ultimate} and
\textit{Quickest-increment}, respectively. In other words,
$\lambda_i$ and $\mu_i$ are the optimal values of the problems
\textsc{MinArcs}($B,i+1$) which are solved for increasing the flow
from $f+i$ to $f+i+1$. Furthermore, let $c_j$ for $0\leqslant
j\leqslant r$ be the minimum number of potential arcs that have to
be built in order to reach a flow of $f+j$:
\begin{align}
\nonumber c_j &= \min\{ \sum_{a\in A_{p}}x_a\quad :\quad \sum_{a\in\delta^{\text{out}}(v)}x_a-\sum_{a\in\delta^{\text{in}}v)}x_a=
  \begin{cases}
    0 & \text{for }v\in V\setminus\{s,t\}, \\
    f+j & \text{for }v=s,\\
    -f-j & \text{for }v=t,
  \end{cases}\\
&\qquad\qquad\qquad\qquad 0\leqslant x_a\leqslant 1\text{ for all }a\in A\} \label{eq:c_primal}\\
\nonumber &= \max\{ (f+j)(\pi_s-\pi_t)+\sum_{a\in A}y_a\quad :\quad \pi_v-\pi_w+y_a\leqslant
\begin{cases}
  1 & \text{for }a=(v,w)\in A_p, \\
  0 & \text{for }a=(v,w)\in A_e,
\end{cases} \\
&\qquad\qquad y_a\leqslant 0\text{ for all }a\in A\}. \label{eq:c_dual}
\end{align}
Observe that if the associated sets of arcs would be nested, then
they would give rise to an optimal solution. However, since this is
unlikely to be the case in general, the values $c_j$ can only be
used to derive an upper bound.
\begin{lemma}\label{lem:bound_opt_by_c}
  $\displaystyle z^*\leqslant TF-\sum_{j=1}^rc_j\leqslant TF-\frac12r(r-1)-c_r$.
\end{lemma}
\begin{proof}
For $i=0,\ldots,r-1$, let $\lambda^*_i$ be the number of time
periods with flow $f+i$ in an optimal solution. Since any feasible solution must have at least
$c_{i+1}$ periods with flow value at most $f+i$, we have
\[\lambda^*_0+\cdots+\lambda^*_i\geqslant c_{i+1}.\]
Summing over $i$, we obtain
\[\sum_{i=0}^{r-1}\lambda^*_i(r-i)\geqslant\sum_{j=1}^rc_j,\]
hence
\begin{multline*}
z^*=\sum_{i=0}^{r-1}\lambda^*_i(f+i)+\left(T-\sum_{i=0}^{r-1}\lambda^*_i\right)(f+r)=TF-\sum_{i=0}^{r-1}\lambda^*_i\left[(f+r)-(f+i)\right]\\
=TF-\sum_{i=0}^{r-1}\lambda^*_i(r-i)\leqslant TF-\sum_{j=1}^rc_j.
\end{multline*}
This proves the first inequality and the second one follows with
$c_j\geqslant j$ for $1\leqslant j\leqslant r$.
\end{proof}
The values $z_1$ and $z_2$ can be expressed in terms of the
sequences $(\lambda_i)$ and $(\mu_i)$, respectively.
\begin{lemma}\label{lem:z_i_by_lambda_resp_mu}
$\displaystyle z_1=TF-\sum_{i=0}^{r-1}\lambda_i(r-i)$ and $\displaystyle z_2=TF-\sum_{i=0}^{r-1}\mu_i(r-i)$.
\end{lemma}
\begin{proof}
We have
\[z_1=\sum_{i=0}^{r-1}\lambda_i(f+i)+\left(T-\sum_{i=0}^{r-1}\lambda_i\right)(f+r)=TF-\sum_{i=0}^{r-1}\lambda_i(r-i),\]
and similarly for $z_2$.
\end{proof}
The sequences $(\lambda_i)_{i=0,\ldots,r-1}$ and
$(\mu_i)_{i=0,\ldots,r-1}$ are non-decreasing.
\begin{lemma}\label{lem:lambda_incr}
We have
$\lambda_0\leqslant\lambda_1\leqslant\cdots\leqslant\lambda_{r-1}$
and $\mu_0\leqslant\mu_1\leqslant\cdots\leqslant\mu_{r-1}$.
\end{lemma}
\begin{proof}
Let $\overline A=A_e\cup\overline A_p$ be the set of arcs that are
used by \textit{Quickest-to-ultimate}. Set $B_0=\emptyset$, and
for $i=1,\ldots,r-1$ let $B_i$ be the set of arcs which are built by
\textit{Quickest-to-ultimate} in the first
$\lambda_0+\cdots+\lambda_{i-1}$ time periods. Then
\begin{align*}
\lambda_i &= \min\{ \sum_{a\in\overline A_{p}\setminus B_i}x_a\quad :\quad \sum_{a\in\delta^{\text{out}}(v)}x_a-\sum_{a\in\delta^{\text{in}}(v)}x_a=
  \begin{cases}
    0 & \text{for }v\in V\setminus\{s,t\}, \\
    f+i+1 & \text{for }v=s,\\
    -f-i-1 & \text{for }v=t,
  \end{cases}\\
&\qquad\qquad\qquad\qquad 0\leqslant x_a\leqslant 1\text{ for all }a\in \overline A\} \\
&= \max\{ (f+i+1)(\pi_s-\pi_t)+\sum_{a\in\overline A}y_a\ :\ \pi_v-\pi_w+y_a\leqslant
\begin{cases}
  1 & \text{for }a=(v,w)\in \overline A_p\setminus B_i, \\
  0 & \text{for }a=(v,w)\in A_e\cup B_i,
\end{cases} \\
&\qquad\qquad y_a\leqslant 0\text{ for all }a\in \overline A\}.
\end{align*}
Fix some $i<r-1$, and let $(\pi^*,y^*)$ be an optimal solution for
the dual characterization of $\lambda_i$. Since $A_e\cup B_i$ can
carry an $s$-$t$-flow of value $f+i$, the optimal value for the optimization problem which
characterizes $\lambda_i$ becomes $0$ when the flow value $f+i+1$ is replaced by $f+i$. In the dual
problem this is corresponds to the coefficient of $(\pi_s-\pi_t)$ in the objective function, and
since $(\pi^*,y^*)$ is also feasible for this modified dual problem, we have
\[(f+i)(\pi^*_s-\pi^*_t)+\sum_{a\in\overline A}y_a\leqslant 0,\]
and therefore
\[\lambda_i=(f+i+1)(\pi^*_s-\pi^*_t)+\sum_{a\in\overline A}y^*_a\leqslant\pi^*_s-\pi^*_t.\]

On the other hand, we obtain a feasible solution $(\pi',y')$ for the
dual characterization of $\lambda_{i+1}$ by setting $\pi'=\pi^*$ and
\[y'_a=
\begin{cases}
  y^*_a-1 & \text{for }a\in B_{i+1}\setminus B_i,\\
  y^*_a & \text{for }a\in \overline A\setminus(B_{i+1}\setminus B_i),
\end{cases}\]
and together with $\lvert B_{i+1}\setminus B_i\rvert=\lambda_i$ this
implies
\[\lambda_{i+1}\geqslant(f+i+2)(\pi^*_s-\pi^*_t)+\sum_{a\in\overline A}y^*_a-\lambda_i=\pi^*_s-\pi^*_t.\]
The inequality $\mu_i\leqslant\mu_{i+1}$ is proved in the same way (with $A$ instead of $\overline A$).
\end{proof}
\begin{theorem}\label{thm:approx_heur_1}
\textit{Quickest-to-ultimate} is a $2$-approximation algorithm for
the incremental maximum flow problem with unit capacities, i.e.,
$z^*\leqslant 2z_1$.
\end{theorem}
\begin{proof}
Using Lemma~\ref{lem:lambda_incr}, we have
\[\sum_{i=0}^{r-1}\lambda_i(r-1-2i)=\sum_{i=0}^{\lfloor(r-1)/2\rfloor}(\lambda_i-\lambda_{r-1-i})(r-1-2i)\leqslant 0.\]
Adding the inequality $(\lambda_0+\cdots+\lambda_{r-1})r\leqslant TF$, it follows that
\[TF-\sum_{i=0}^{r-1}\lambda_i(2(r-i)-1)\geqslant 0.\]
By construction, $c_r=\lambda_0+\cdots+\lambda_{r-1}$, and with
Lemma~\ref{lem:bound_opt_by_c} we derive
\begin{multline*}
  z^*\leqslant TF - \frac12r(r-1)-\sum_{i=0}^{r-1}\lambda_i\leqslant TF - \frac12r(r-1)-\sum_{i=0}^{r-1}\lambda_i+\left(TF-\sum_{i=0}^{r-1}\lambda_i(2(r-i)-1)\right)\\
=2\left(TF-\sum_{i=0}^{r-1}\lambda_i(r-i)\right)-\frac12r(r-1)=2z_1-\frac12r(r-1).\qedhere
\end{multline*}
\end{proof}

In order to prove that for \textit{Quickest-increment} the case in
Figure~\ref{fig:example_2} illustrates an upper bound for the
approximation ratio, we need some technical preparations. We start
by bounding the numbers $\mu_i$ by the numbers $c_j$.
\begin{lemma}\label{lem:mu_bound_by_c}
We have $\displaystyle\mu_i\leqslant c_j/(j-i)$ for all $(i,j)$ with
$0\leqslant i<j\leqslant r$.
\end{lemma}
\begin{proof}
Fix $i$ and $j$ with $0\leqslant i<j\leqslant r$, and let $X$ be the feasible region of the maximization problem which characterizes $\mu_i$, i.e.,
\begin{multline*}
X=\Big\{(y,\pi)\in\reals^{\lvert A\rvert}\times\reals^{\lvert V\rvert}\ :\ \pi_v-\pi_w+y_a\leqslant
  \begin{cases}
    0 & \text{for }a=(v,w)\in A_e\cup B_i,\\
    1 & \text{for }a=(v,w)\in A_p\setminus B_i,
  \end{cases} \\
y_a\leqslant 0\text{ for all }a\in A\Big\},
\end{multline*}
where $B_i$ is the set of arcs built by \textit{Quickest-increment}
in the first $\mu_0+\cdots+\mu_{i-1}$ time periods. Since the
maximum flow value on $A_e\cup B_i$ equals $f+i$, we have that
$X\subseteq Y$ where
\begin{multline*}
Y=\Big\{(y,\pi)\in\reals^{\lvert A\rvert}\times\reals^{\lvert V\rvert}\ :\ \pi_v-\pi_w+y_a\leqslant
  \begin{cases}
    0 & \text{for }a=(v,w)\in A_e,\\
    1 & \text{for }a=(v,w)\in A_p,
  \end{cases}\\
(f+i)(\pi_s-\pi_t)+\sum_{a\in A}y_a\leqslant 0,\ y_a\leqslant 0\text{
for all }a\in A\Big\},
\end{multline*}
Consequently,
\[\mu_i\leqslant \max\Big\{(f+i+1)(\pi_s-\pi_t)+\sum_{a\in A}y_a\ :\ (y,\pi)\in Y\Big\},\]
and by taking the dual we obtain
\begin{multline}\label{eq:aux_problem}
\mu_i\leqslant\min\Big\{ \sum_{a\in A_{p}}x_a\ :\ \sum_{a\in\delta^{\text{out}}(v)}x_a-\sum_{a\in\delta^{\text{in}}v)}x_a=
\begin{cases}
0 & \text{for }v\in V\setminus\{s,t\}, \\
f+i+1-(f+i)z & \text{for }v=s, \\
-f-i-1+(f+i)z & \text{for }v=t,
\end{cases}\\
x_a+z\leqslant 1\text{ for all }a\in A,\quad x_a \geqslant 0\text{
for all }a\in A,\ z\geqslant 0\Big\}.
\end{multline}
Let $x\in\reals^{\lvert A\rvert}$ be an optimal solution for the
problem~(\ref{eq:c_primal}) to determine $c_j$, i.e., an
$s$-$t$-flow of value $f+j$. Solving the system
\begin{align*}
(f+j)\gamma &= f+i+1-(f+i)z, & \gamma+z &= 1
\end{align*}
for $\gamma$ and $z$ yields $\gamma=1/(j-i)$ and $z=(j-i-1)/(j-i)$,
and this implies that a feasible solution $(x',z)$ for
problem~(\ref{eq:aux_problem}) can be obtained by setting
$x'_a=x_a/(j-i)$ for all $a\in A$ and $z=(j-i-1)/(j-i)$. Hence the
optimal value of problem~(\ref{eq:aux_problem}) is at most
$c_j/(j-i)$ and this concludes the proof.
\end{proof}
From Lemmas~\ref{lem:lambda_incr} and~\ref{lem:mu_bound_by_c} it
follows that $(\mu,c,T)\in\reals^r\times\reals^{r+1}\times\reals$
satisfies the following conditions:
\begin{align}
\mu_i-\mu_{i+1} &\leqslant 0 && \text{for }0\leqslant i\leqslant r-2, \label{eq:monotonicity} \\
(j-i)\mu_i-c_j &\leqslant 0 && \text{for }0\leqslant i<j\leqslant r, \label{eq:lambda_bound}\\
\sum_{i=0}^{r-1}\mu_i -T &\leqslant 0, \label{eq:time_bound}\\
c_0 &= 0, \label{eq:boundary}\\
\mu_i,\,c_j &\geqslant 0 && \text{for }0\leqslant i\leqslant r-1,\, 1\leqslant j\leqslant r. \label{eq:nonnegativity}
\end{align}
Let $X(r)\subseteq\reals^r\times\reals^{r+1}\times\reals$
denote the set of all $(\mu,c,T)$ satisfying~(\ref{eq:monotonicity})
to~(\ref{eq:nonnegativity}). Using Lemmas~\ref{lem:bound_opt_by_c}
and~\ref{lem:z_i_by_lambda_resp_mu}, a number $\gamma>1$ is an upper
bound for the approximation ratio of \textit{Quickest-increment},
provided
\[TF-\sum_{j=1}^rc_j\leqslant \gamma\left[TF-\sum_{i=0}^{r-1}\mu_i(r-i)\right]\]
is a valid inequality for $X(r)$, or equivalently
\[F(1-\gamma)T-\sum_{j=1}^rc_j+\gamma\sum_{i=0}^{r-1}\mu_i(r-i)\leqslant 0\]
for all $(\mu,c,T)\in X(r)$. Since  $F\geqslant r$ and $\gamma>1$,
this inequality is strengthened when $F$ is replaced by $r$ on the
left hand side. By duality, the stronger inequality is valid for
$X(r)$, provided the following system has a solution:
\begin{align}
  \sum_{i=0}^{j-1}y_{ij} &\leqslant 1 && \text{for }1\leqslant j\leqslant r, \label{eq:dual_for_c}\\
  z+\sum_{j=i+1}^r(j-i)y_{ij}+x_i-x_{i-1} &\geqslant\gamma(r-i) && \text{for }0\leqslant i\leqslant r-1, \label{eq:dual_for_lambda}\\
  z &\leqslant r(\gamma-1),\label{eq:dual_for_T}\\
  x_{-1}=x_{r-1} &=0, \label{eq:boundaries}\\
  z,x_i &\geqslant 0 && \text{for }0\leqslant i\leqslant r-2, \label{eq:nonoeg_1}\\
  y_{ij} &\geqslant 0 && \text{for }0\leqslant i<j\leqslant r. \label{eq:nonneg_2}
\end{align}
Let
$Y(r,\gamma)\subseteq\reals^{r+1}\times\reals^{r(r+1)/2}\times\reals$
be the set of all $(x,y,z)$ satisfying~(\ref{eq:dual_for_c})
to~(\ref{eq:nonneg_2}). The above argument implies the following
theorem which provides a sufficient condition for $\gamma$ to be an
approximation factor for \textit{Quickest-increment} on all
instances with $F-f=r$.
\begin{theorem}\label{thm:suff_cond}
If $Y(r,\gamma)\neq\emptyset$, then \textit{Quickest-increment} is
a $\gamma$-approximation algorithm for all instances of the
incremental maximum flow problem with unit capacities and
$F-f=r$.\hfill\qed
\end{theorem}
Consequently, in order to establish that \textit{Quickest-increment}
is a $3/2$-approximation algorithm, it is sufficient to prove the
$Y(r,3/2)\neq\emptyset$ for all $r\geqslant 2$. To do this, we
need another lemma.
\begin{lemma}\label{lem:initial_segment}
Let $r\geqslant 2$ and $s=\lfloor(r+2)/3\rfloor$. Then the following
system has a solution with non-negative $y_{ij}$ for $0\leqslant
i\leqslant s$, $i+1\leqslant j\leqslant r$:
  \begin{align*}
    \sum_{j=i+1}^{r}(j-i)y_{ij} &\geqslant r-\frac32i &&0\leqslant i\leqslant s,\\
    \sum_{i=0}^{\min\{j-1,s\}}y_{ij} &\leqslant 1 &&1\leqslant j\leqslant r.
  \end{align*}
\end{lemma}
\begin{proof}
We proceed by induction on $r$. For $r=2$, $s=1$, a solution is
given by $y_{01}=1$, $y_{02}=y_{12}=1/2$. So we
assume $r>2$, and we distinguish two cases.
\begin{description}
\item[Case 1.] $r\equiv 0\text{ or }2\pmod 3$. By induction, there are nonnegative numbers $y_{ij}$ for $0\leqslant i\leqslant s$ and $i+1\leqslant j\leqslant r-1$ satisfying
  \begin{align*}
    \sum_{j=i+1}^{r-1}(j-i)y_{ij} &\geqslant r-1-\frac32i &&0\leqslant i\leqslant s,\\
    \sum_{i=0}^{\min\{j-1,s\}}y_{ij} &\leqslant 1 &&1\leqslant j\leqslant r-1.
  \end{align*}
We extend this by $y_{ir}=1/(r-i)$ for $0\leqslant i\leqslant s$.
Then $\sum_{j=i+1}^{r}(j-i)y_{ij} \geqslant r-\frac32i$ for
$0\leqslant i\leqslant s$ and
\[\sum_{i=0}^sy_{ir}=\sum_{i=0}^s\frac{1}{r-i}\leqslant\frac{s+1}{r-s}\leqslant 1.\]
\item[Case 2.] $r\equiv 1\pmod 3$ and $s=(r+2)/3$. For $r=4$, a solution is given by
\begin{align*}
y_{01} &= 1, & y_{02} &= 1, & y_{03} &= 1/3, & y_{04} &= 0,\\
 && y_{12} &= 0, & y_{13} &= 2/3 , & y_{14} &= 7/18,\\
 && & & y_{23}&=0, & y_{24} &= 1/2.
\end{align*}
For $r>4$, by induction, there are nonnegative numbers $y_{ij}$ for
$0\leqslant i\leqslant s-1$ and $i+1\leqslant j\leqslant r-3$
satisfying
\begin{align*}
    \sum_{j=i+1}^{r-3}(j-i)y_{ij} &\geqslant r-3-\frac32i &&0\leqslant i\leqslant s-1,\\
    \sum_{i=0}^{\min\{j-1,s-1\}}y_{ij} &\leqslant 1 &&1\leqslant j\leqslant r-3.
\end{align*}
Let $i_0=\lfloor(r-5)/4\rfloor$. We extend our solution by
\begin{align*}
y_{ij} &=
\begin{cases}
\frac{3}{(r-2)-i} &\text{for }j=r-2,\\
0 &\text{for }j\in\{r-1,r\}
\end{cases}
 &&\text{for }0\leqslant i\leqslant i_0,\\
y_{ij} &=
\begin{cases}
\frac{3}{(r-1)-i} &\text{for }j=r-1,\\
0 &\text{for }j\in\{r-2,r\}
\end{cases}
 &&\text{for }i_0+1\leqslant i\leqslant s-1,\\
y_{sj} &=
\begin{cases}
  \frac{3(r-2)}{4(r-1)} &\text{for }j=r,\\
0 &\text{for }j<r.
\end{cases}
\end{align*}
Then $\sum_{j=i+1}^{r}(j-i)y_{ij} \geqslant r-\frac32i$ for
$0\leqslant i\leqslant s$,
\[\sum_{i=0}^sy_{i,r-2}=\sum_{i=0}^{i_0}\frac{3}{(r-2)-i}\leqslant\frac{3(i_0+1)}{r-2-i_0}\leqslant\frac{3(r-1)/4}{r-2-(r-5)/4}=1,\]
\begin{multline*}
\sum_{i=0}^sy_{i,r-1}=\sum_{i=i_0+1}^{s-1}\frac{3}{(r-1)-i}\leqslant\frac{3(s-1-i_0)}{r-1-(s-1)}\leqslant\frac{3\left(\frac{r+2}{3}-1-\frac{r-5}{4}\right)}{r-\frac{r+2}{3}}\\
=\frac{\frac14\left(r+11\right)}{\frac13(2r-2)}=\frac{3r+33}{8r-8}\leqslant 1,
\end{multline*}
where we used $r\geqslant 5$ for the last inequality, and
$\sum_{i=0}^sy_{ir}=y_{sr}\leqslant 1$.\qedhere
\end{description}
\end{proof}
Using Lemma~\ref{lem:initial_segment}, we can exhibit a point in
$Y(r,3/2)$.
\begin{lemma}\label{lem:witness_point}
  $Y(r,3/2)\neq\emptyset$ for all $r\geqslant 2$.
\end{lemma}
\begin{proof}
Let $s=\lfloor(r+2)/3\rfloor$ and let $y_{ij}$ for $0\leqslant
i\leqslant s$ and $i+1\leqslant j\leqslant r$ be a point as
described in Lemma~\ref{lem:initial_segment}. For $i>s$ and
$i+1\leqslant j\leqslant r$, let $y_{ij}=0$. In addition, let
$x_i=0$ for $0\leqslant i\leqslant s$, and $z=r/2$. Finally, let
\begin{align*}
x_i &= (i-s)\left(r-\frac32s-\frac34(i-s+1)\right) &\text{for }s+1\leqslant i\leqslant r-2.
\end{align*}
We claim that $(x,y,z)\in Y(r,3/2)$.
Constraints~(\ref{eq:dual_for_c}) follow immediately
Lemma~\ref{lem:initial_segment}. The same is true for
constraints~(\ref{eq:dual_for_lambda}) for $0\leqslant i\leqslant
s$. For $s<i\leqslant r-2$, constraints~(\ref{eq:dual_for_lambda})
are satisfied because
\[x_i-x_{i-1}=(i-s)\left(r-\frac32s-\frac34(i-s+1)\right)-(i-1-s)\left(r-\frac32s-\frac34(i-s)\right)=r-\frac32i.\]
Finally, constraint~(\ref{eq:dual_for_lambda}) for $i=r-1$ follows from
\[x_{r-2}=(r-2-s)\left(r-\frac32s-\frac34(r-1-s)\right)=
\begin{cases}
  r/2-3/2 & \text{if }r\equiv 0\pmod 3,\\
  r/6-2/3 & \text{if }r\equiv 1\pmod 3,\\
  r/3-7/6 & \text{if }r\equiv 2\pmod 3,\\
\end{cases}\]
which also implies $x_i\geqslant 0$ for all $i$.
\end{proof}
The bound of $3/2$ for the approximation ratio of
\textit{Quickest-increment} is a consequence of
Theorem~\ref{thm:suff_cond} and Lemma~\ref{lem:witness_point}.
\begin{theorem}\label{thm:approx_heur_2}
\textit{Quickest-increment} is a $3/2$-approximation algorithm for
the incremental maximum flow problem with unit capacities.
\end{theorem}

Let us examine the example in Figure~\ref{fig:example_2} more
closely. It has initial flow value $f = 0$ and ultimate flow value
$F = 2$ and the essence of the example is that the arcs that are
built to reach a flow value 1 cannot be used for the (ultimate) flow
value 2. Figure \ref{fig:r_3} shows an attempt to generalize the
example to one with initial flow value $f = 0$ and ultimate flow
value $F = 3$. The instance has a time horizon $ T = 7k + 3$.
\begin{figure}[htb]
  \centering
\includegraphics[width=.5\textwidth]{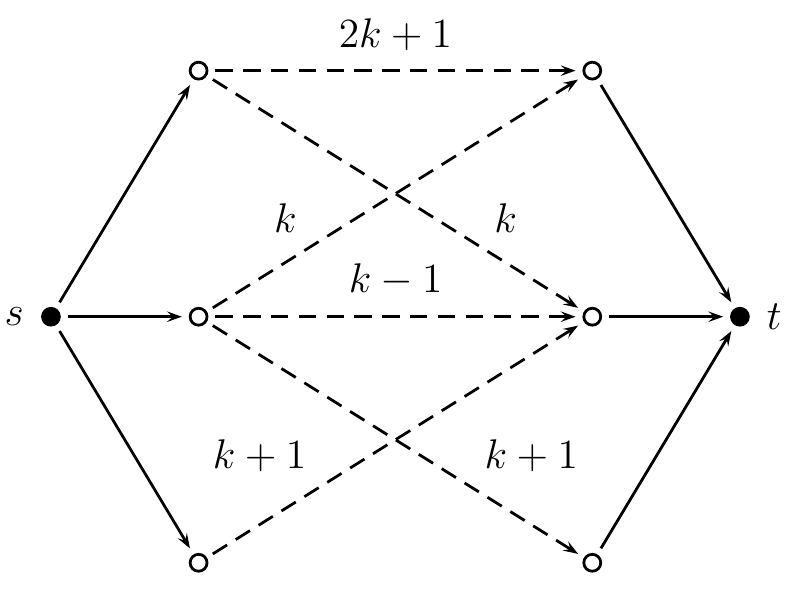}
  \caption{An instance with $r=3$.}
  \label{fig:r_3}
\end{figure}
\textit{Quickest-increment} starts with building the middle path
with $k-1$ potential arcs. Then the cheapest way to achieve a flow
of value 2 is to build the two paths containing $k$ potential arcs
each. Finally, all the remaining arcs have to be built in order to
achieve a flow of value 3. This gives a total flow of
\[
(k-1) \cdot 0 + 2k \cdot 1 + (4k+3) \cdot 2 + 3 = 10k + 9.
\]
On the other hand, the optimum for large $k$ is to build the
ultimate flow immediately which yields
\[
(k+1) \cdot 0 + (k+1) \cdot 1 + (2k+1) \cdot 2 + 3k \cdot 3 = 14k +
3.
\]
What works in favor of the heuristic, in terms of performance, is
that in order for the heuristic not to build the upper path to
achieve a flow of value 2, it has to be long enough, i.e. longer than
the sum of the lengths of the diagonal paths that the heuristic
chooses.

If we define
\[
\varphi(r) = \inf \{ \gamma \ : \ \text{\textit{Quickest-increment}
is a $\gamma$-approximation for instances with } F-f = r \},
\]
then Theorem \ref{thm:approx_heur_2} shows that $\varphi(r)
\leqslant 3/2$. It is possible to show that $\varphi(r) \leqslant
4/3$ for $r \geqslant 70$, and we conjecture that it is possible to
show that $\lim_{r \to \infty} \varphi(r) = 1$.

\subsection{The incremental maximum matching problem}\label{subsec:inc_max_matching}

In our quest for polynomially solvable cases of incremental network
design with maximum flows, we next consider, in addition to unit
capacities on all arcs, imposing a restriction on the structure of
the network. More specifically, a node set $N = V \cup W \cup \{s,
t\}$ and an arc set $A$ consisting of existing arcs $(s,v)$ for all
$v \in V$, existing arcs $(w,t)$ for all $w \in W$, and some arcs
$(v, w)$ with $v \in V$ and $w \in W$, which can be either existing
or potential. Thus, $s$-$t$ flows correspond to matchings in the
bipartite graph $(V \cup W, \{ \{v,w\} : (v,w)\in A\})$, and
therefore we call this special case the \emph{incremental maximum
matching problem}.

The example in Figure \ref{fig:h1-matching} shows that
\textit{Quickest-to-ultimate} may fail to find an optimal solution.
The unique maximum cardinality matching is $\{(1,2), (3,4), \ldots,
(15,16)\}$ and building arcs (1,2), (3,4), (5,6), and (7,8),
followed by arcs (9,10), (11,12), (13,14), and (15,16), followed by
(9,4) results in a total cumulative flow for
\textit{Quickest-to-ultimate} of $4 \cdot 6 + 4 \cdot 7 + 2 \cdot 8
= 68$, whereas building arcs (1,2) and (9,4), followed by arcs
$(3,4),(5,6),\ldots,(15,16)$ results in a total cumulative flow of
$2\cdot 6+7\cdot 7+8=69$.
\begin{figure}[ht]
\centering
\includegraphics[width=.6\textwidth]{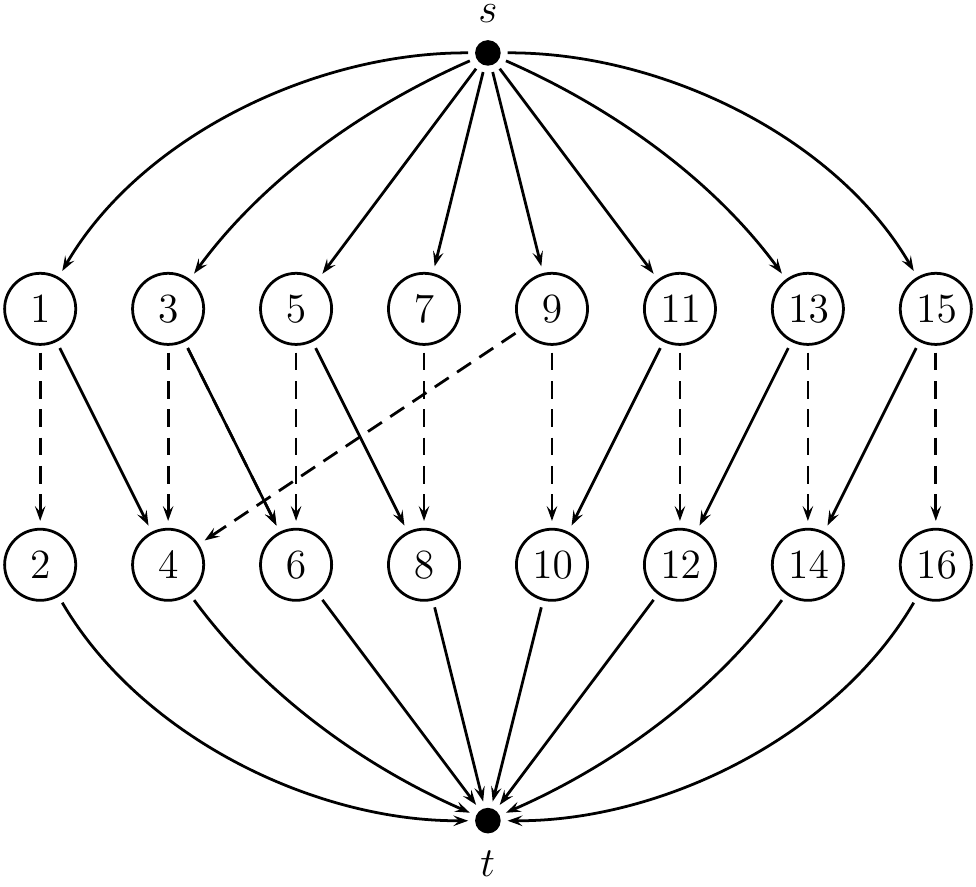}
\caption{Bad instance for
\textit{Quickest-to-ultimate}.}\label{fig:h1-matching}
\end{figure}

Similarly, the example in Figure~\ref{fig:h2-matching} shows that
\textit{Quickest-increment} may fail to find an optimal solution.
The initial maximum matching has size $5$, and the unique cheapest
way to obtain a matching of size 6 is to build arcs $(5,8)$ and
$(7,10)$. After that, all the remaining potential arcs need to be
build to reach a matching of size 7, and the total cumulative flow
for \textit{Quickest-increment} is $2 \cdot 5 + 6 \cdot 6 + 7 = 53$.
On the other hand, building arcs $(1,2), \ldots,(13,14)$ followed by
arcs (5,8) and (7,10) results in a total cumulative flow of $3 \cdot
5 + 3 \cdot 6 + 3 \cdot 7 = 54$.
\begin{figure}[ht]
  \centering
\includegraphics[width=.6\textwidth]{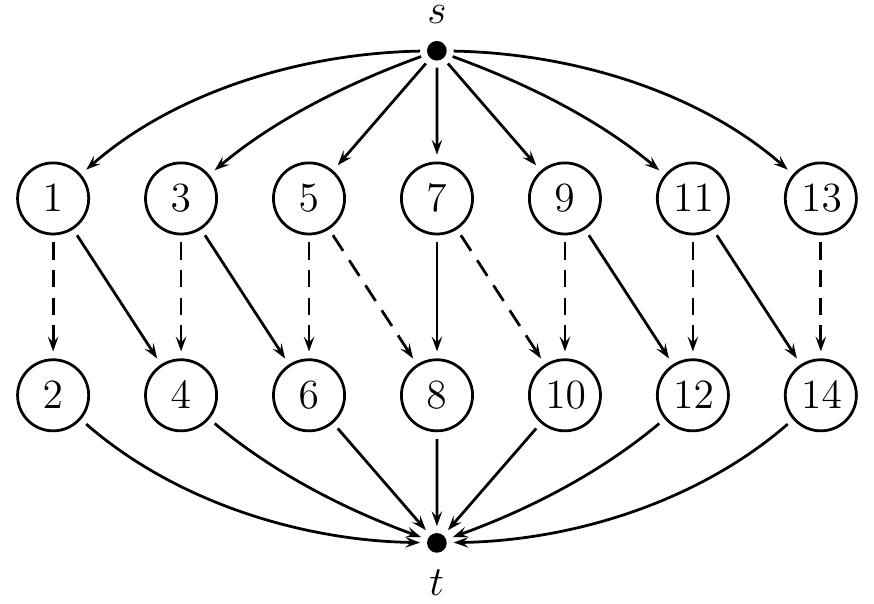}
\caption{Bad instance for
\textit{Quickest-increment}.}\label{fig:h2-matching}
\end{figure}

Thus, even restricting the structure of the network to bipartite
graphs does not trivially give a polynomially solvable case of
incremental network design with maximum flows. We have been able to
improve the performance guarantee of \textit{Quickest-to-ultimate}
to $4/3$ for instances of the incremental maximum matching problem.
The proof is given in the appendix.

\section{Computational study}\label{sec:computations}

We have conducted a computational study to gain insight into the
characteristics that make instances of the incremental maximum flow
problem hard to solve, to compare the performance of the two MIP
formulations, and to assess the performance of the three heuristics.

The following computational tools were used to develop and analyze
the formulations and algorithms: Python 2.7.5, Matplotlib 1.3.1,
NetworkX 1.8.1, and Gurobi 5.5. All computational experiments were
conducted on a 64-bit Win7 with Intel Xeon CPU (E5-1620) with 8
cores and 32GB of RAM using a single thread.

\subsection{Instance generation}\label{subsec:instances}

For the computational experiments in which we compare the
performance of the IP formulations and the heuristics, we use two
classes of instances: one using general graphs and one using layered
graphs.

General graphs are parameterized by the number of nodes $n$, the
expected density $d = 2 |A| / n(n-1) \in \{0.1, 0.3, 0.7\}$, the
expected fraction of potential arcs $p = |A_p|/|A| \in \{0.3,
0.7\}$, and the maximum arc capacity $u_{\max} \in \{1, 3, 10\}$. For
each combination of parameters, 10 random instances are generated.
More specifically, for a given number of nodes $n$, an arc between
two nodes exists with probability $d$, and if an arc between two
nodes exists, the arc is a potential arc with probability $p$ (and
thus an existing arc with probability $1-p$).

Layered graphs consist of a source node, a sink node, and $\ell$
layers in between $(\ell \geqslant 2)$. The source is connected to every
node in the first layer and every node in the last layer is
connected to the sink. The nodes in layer $i$ may be connected to a
node in layer $i+1$. Layered graphs are parameterized by the number
of layers $\ell$, the number of nodes in each layer $n$, the
expected density $d = |A|/(\ell-1)n^2 \in \{0.1, 0.3, 0.7\}$, the
expected fraction of potential arcs $p = |A_p|/|A| \in \{0.3,
0.7\}$, and the maximum arc capacity $u_{\max} \in \{1, 3, 10\}$. For
each combination of parameters, 10 random instances are generated.
The generation proceeds similar to the generation of general
instances. Layered instances are considered because they ensure a
minimum distance between the source and the sink.

\subsection{Comparison of the MIP models}\label{subsec:mip_computations}

We start by determining the characteristics that make an instance
difficult to solve and, at the same time, comparing the performance
of the two MIP formulations presented in Section~\ref{sec:problem}.

Tables \ref{tbl:mip_gen_ave} and \ref{tbl:mip_layered_ave} present
average performance statistics for various combinations of the
instance parameters (with $n=35$ for general graphs and $\ell = 5$
and $n = 10$ for layered graphs). More specifically, we report the
arc density ($d$), the fraction of potential arcs ($p$), the maximum
arc capacity ($u_{\max}$), the difference between the value of an
ultimate and an initial maximum flow ($F - f$), the number of
instances not solved to optimality within the time limit (\#), the
average initial gap (initial gap), computed as
$(z^{\text{LP}}-z^{\text{IP}})/z^{\text{IP}}$, where $z^{\text{LP}}$ the value of the LP
relaxation and $z^{\text{IP}}$ is the value of the best solution found, the
average final gap over the instances not solved to optimality (final
gap), where in the gap computation $z^{\text{LP}}$ is replaced by the best
bound at termination, and the average solution time (time) and
average number of nodes in the search tree (\#nodes) over the
instances solved to optimality.

\begin{table}[htdp]
\caption{Comparison of MIP models for general graphs with $n = 35$;
time limit 1 hour} \footnotesize
\begin{center}
\begin{tabular}{@{} rrrrrrrrrrrrrr @{}}
\toprule
\multicolumn{4}{c}{} & \multicolumn{5}{c}{IMFP$^1$} & \multicolumn{5}{c}{IMFP$^2$} \\
\cmidrule(lr){5-9}\cmidrule(lr){10-14}
    &     &           &       &  & initial  & final    &       &         &  & initial  & final    &        &          \\
$d$ & $p$ & $u_{\max}$ & $F-f$ &\#& gap      & gap      & time  & \#nodes &\#& gap      & gap      & time   & \#nodes \\
\midrule
    0.1   & 0.3   & 1     & 0.9   &       & 0.12\% &       & 0.13  & 1.5   &       & 0.00\% &       & 0.01  & 0.0 \\
    0.1   & 0.3   & 3     & 1.5   &       & 0.28\% &       & 0.15  & 0.5   &       & 0.43\% &       & 0.00  & 0.0 \\
    0.1   & 0.3   & 10    & 7.3   &       & 0.36\% &       & 0.22  & 1.5   &       & 0.94\% &       & 0.05  & 4.5 \\
\addlinespace
    0.1   & 0.7   & 1     & 2.1   &       & 0.55\% &       & 23.88 & 309.1 &       & 0.00\% &       & 0.00  & 0.0 \\
    0.1   & 0.7   & 3     & 4.0   &       & 1.17\% &       & 7.33  & 325.9 &       & 1.11\% &       & 0.03  & 0.0 \\
    0.1   & 0.7   & 10    & 11.6  &       & 1.46\% &       & 73.66 & 1,694.7 &       & 1.83\% &       & 1.73  & 107.1 \\
\addlinespace
    0.3   & 0.3   & 1     & 1.9   &       & 0.02\% &       & 1.56  & 75.5  &       & 0.00\% &       & 0.01  & 0.0 \\
    0.3   & 0.3   & 3     & 7.7   &       & 0.07\% &       & 3.01  & 66.0  &       & 0.13\% &       & 0.05  & 0.0 \\
    0.3   & 0.3   & 10    & 11.7  &       & 0.03\% &       & 1.81  & 7.0   &       & 0.13\% &       & 0.09  & 0.0 \\
\addlinespace
    0.3   & 0.7   & 1     & 6.4   &       & 0.13\% &       & 288.94 & 1,461.1 &       & 0.00\% &       & 0.04  & 0.0 \\
    0.3   & 0.7   & 3     & 12.8  & 1     & 0.14\% & 0.02\% & 594.46 & 3,529.1 &       & 0.17\% &       & 0.46  & 0.7 \\
    0.3   & 0.7   & 10    & 37.0  & 1     & 0.14\% & 0.10\% & 607.91 & 925.4 & 1     & 0.28\% & 0.04\% & 9.91  & 73.6 \\
\bottomrule
\end{tabular}
\end{center}
\label{tbl:mip_gen_ave}
\end{table}%

\begin{table}[htdp]
\caption{Comparison of MIP models for layered graphs with $\ell = 5$
and $n = 10$; time limit 1 hour.} \footnotesize
\begin{center}
\begin{tabular}{@{} rrrrrrrrrrrrrr @{}}
\toprule
\multicolumn{4}{c}{} & \multicolumn{5}{c}{IMFP$^1$} & \multicolumn{5}{c}{IMFP$^2$} \\
\cmidrule(lr){5-9}\cmidrule(lr){10-14}
    &     &           &       &  & initial  & final    &       &         &  & initial  & final    &        &          \\
$d$ & $p$ & $u_{\max}$ & $F-f$ &\#& gap      & gap      & time  & \#nodes &\#& gap      & gap      & time   & \#nodes \\

\midrule
    0.1   & 0.3   & 1     & 1.6   &       & 2.13\% &       & 0.02 & 6.3   &       & 0.00\% &       & 0.00  & 0.0 \\
    0.1   & 0.3   & 3     & 1.3   &       & 1.67\% &       & 0.01 & 0.0     &       & 0.99\% &       & 0.00  & 0.0 \\
    0.1   & 0.3   & 10    & 6.6   &       & 4.76\% &       & 0.03 & 8.1   &       & 6.46\% &       & 0.01  & 0.0 \\
\addlinespace
    0.1   & 0.7   & 1     & 3.1   &       & 2.84\% &       & 1.35 & 502.9 &       & 0.00\% &       & 0.00  & 0.0 \\
    0.1   & 0.7   & 3     & 4.0   &       & 7.20\% &       & 3.24 & 905.5 &       & 6.79\% &       & 0.01  & 0.0 \\
    0.1   & 0.7   & 10    & 5.5   &       & 5.69\% &       & 0.16 & 119.8 &       & 6.77\% &       & 0.01  & 0.0 \\
\addlinespace
    0.3   & 0.3   & 1     & 1.7   &       & 0.09\% &       & 0.36  & 3.3   &       & 0.00\% &       & 0.00  & 0.0 \\
    0.3   & 0.3   & 3     & 6.5   &       & 0.86\% &       & 11.23 & 1603.6 &       & 1.00\% &       & 0.08  & 18.7 \\
    0.3   & 0.3   & 10    & 15.5  &       & 0.80\% &       & 3.76 & 418.4 &       & 1.29\% &       & 3.36  & 394.5 \\
\addlinespace
    0.3   & 0.7   & 1     & 8.9   & 7     & 0.72\% & 0.35\% & 1473.03 & 77596.0 &       & 0.00\% &       & 0.03  & 0.0 \\
    0.3   & 0.7   & 3     & 14.8  & 8     & 3.43\% & 0.52\% & 1398.13 & 67872.5 & 1     & 3.30\% & 0.24\% & 42.99 & 1574.8 \\
    0.3   & 0.7   & 10    & 32.0  & 9     & 4.64\% & 0.66\% & 1225.42 & 17924.0 & 9     & 5.07\% & 0.65\% & 1,705.86 & 15948.0 \\
\bottomrule
\end{tabular}
\end{center}
\label{tbl:mip_layered_ave}
\end{table}%

First and foremost, we observe that IMFP$^2$ performs significantly
better than IMFP$^1$; not only are more instances solved to
optimality within the time limit of 1 hour, but the instances are
solved much faster. (Interestingly, the initial gap for IMFP$^2$
tends to be larger than the initial gap for IMFP$^1$ when the
maximum arc capacity is greater than 1.)

As expected, the instances get more difficult when the density, and
thus the number of arcs, increases, when the fraction of potential
arcs, and thus the planning horizon, increases, and when the maximum
arc capacity increases.

IMFP$^1$ is impacted primarily by the length of the planning
horizon, i.e., the number of possible build sequences, whereas
IMFP$^2$ is impacted primarily by the difference between the value
of the ultimate flow $F$ and the value of the initial flow $f$,
i.e., the flow increase that has to be accommodated.

The density $d$ relates to the number of arcs in the graph and thus
the number of paths from source to sink and thus the maximum flow.
Therefore, if the density increases, then the length of the planning
horizon and the difference between $F$ and $f$ increase, which
impacts the performance of both formulations. The fraction of
potential arcs $p$ relates to the planning horizon and thus the
number of possible sequences in which arcs can be build, but also
impacts $f$ and thus the difference between $F$ and $f$. Therefore,
if this fraction increases, then the length of the planning horizon
and the difference between $F$ and $f$ increase, which should impact
the performance of both formulations, but most likely the
performance of IMFP$^1$. The maximum arc capacity $u_{\max}$ impacts
the difference between $F$ and $f$. Therefore, if the arc capacity
increases, then the difference between $F$ and $f$ increases, which
impacts the performance of IMFP$^2$, but also IMFP$^1$.

To further investigate the behavior of the two formulations, Tables
\ref{tbl:mip_gen} and \ref{tbl:mip_layered} present detailed
performance statistics for 10 instances in the most difficult class,
i.e., $d=0.3$, $p=0.7$, and $u_{\max}=10$, when the time limit is
increased to 4 hours. The instances are presented in nondecreasing
order of the difference between the value of the ultimate flow and
the value of the initial flow, i.e., $F - f$.

\begin{table}[htdp]
\caption{Comparison of MIP models for general graphs with $n = 35$,
$d=0.3$, $p=0.7$, $u_{\max}=10$; time limit 4 hours.} \footnotesize
\begin{center}
\begin{tabular}{@{} r rrrr rrrr @{}}
\toprule
 &\multicolumn{4}{c}{IMFP$^1$} & \multicolumn{4}{c}{IMFP$^2$} \\
\cmidrule(lr){2-5}\cmidrule(lr){6-9}
      & initial & final &      &          & initial & final &      &  \\
$F-f$ & gap     & gap   & time & \# nodes & gap     & gap   & time & \# nodes \\
\midrule
15    & 0.08\% & &    48.08 & 34    & 0.16\% & & 0.25  & 0 \\
25    & 0.07\% & &    38.20 & 152   & 0.16\% & & 0.50  & 0 \\
26    & 0.08\% & &   823.75 & 980   & 0.17\% & & 1.72  & 0 \\
38    & 0.24\% & & 3,125.46 & 5,437 & 0.34\% & & 11.13 & 254 \\
39    & 0.20\% & &    83.41 & 520   & 0.36\% & & 1.39  & 0 \\
39    & 0.12\% & &    69.91 & 306   & 0.27\% & & 1.28  & 0 \\
45    & 0.26\% & 0.04\% & 14,400.00 & 40,804 & 0.35\% & & 365.06 & 865 \\
48    & 0.26\% & & 5,989.53 & 23,078 & 0.39\% & & 40.11 & 283 \\
55    & 0.32\% & & 5,134.81 & 24,585 & 0.45\% & & 6.83  & 0 \\
71    & 0.25\% & & 4,333.61 & 30,892 & 0.41\% & & 63.86 & 135 \\
\bottomrule
\end{tabular}
\end{center}
\label{tbl:mip_gen}
\end{table}%

\begin{table}[htdp]
\caption{Comparison of MIP models for layered graphs with $\ell =
5$, $n=10$, $d=0.3$, $p=0.7$, $u_{\max}=10$; time limit 4 hours.}
\footnotesize
\begin{center}
\begin{tabular}{@{} r rrrr rrrr @{}}
\toprule
 &\multicolumn{4}{c}{IMFP$^1$} & \multicolumn{4}{c}{IMFP$^2$} \\
\cmidrule(lr){2-5}\cmidrule(lr){6-9}
      & initial & final &      &          & initial & final &      &  \\
$F-f$ & gap     & gap   & time & \# nodes & gap     & gap   & time & \# nodes \\
\midrule
24    & 2.84\% &        &  3,468.53 & 26,927    & 3.23\% &  & 625.37 & 21,858 \\
26    & 3.82\% & 0.07\% & 14,400.00 & 1,041,170 & 4.27\% &  & 1,450.65 & 102,528 \\
30    & 5.30\% & 0.19\% & 14,400.00 & 210,589   & 5.64\% &  & 2,602.19 & 99,424 \\
34    & 2.85\% & 0.12\% & 14,400.00 & 260,528   & 3.26\% & 0.10\% & 14,400.00 & 212,306 \\
35    & 6.09\% & 0.41\% & 14,400.00 & 354,786   & 6.48\% & 0.26\% & 14,400.00 & 551,579 \\
38    & 4.21\% & 0.10\% & 14,400.00 & 281,519   & 4.67\% &  & 12,281.30 & 196,572 \\
39    & 2.40\% &        & 10,759.84 & 300,278   & 2.83\% & 0.22\% & 14,400.00 & 202,730 \\
39    & 5.24\% & 0.44\% & 14,400.00 & 130,664   & 5.59\% & 0.24\% & 14,400.00 & 312,643 \\
44    & 3.91\% & 0.52\% & 14,400.00 & 37,954    & 4.29\% & 0.42\% & 14,400.00 & 72,215 \\
45    & 3.09\% & 0.32\% & 14,400.00 & 52,449    & 3.64\% & 0.39\% & 14,400.00 & 97,568 \\
\bottomrule
\end{tabular}
\end{center}
\label{tbl:mip_layered}
\end{table}%

We see that the different between the value of the ultimate flow and
the value of the initial flow, i.e., $F - f$, is a good predictor of
the difficulty of an instance for both formulations.  The larger the
value of $F - f$, the more likely it is that the instance cannot be
solved and if an instance cannot be solved, the more likely it is
that the final gap is large. These ``trends'' are most clearly seen
in the results for the layered instances.  For the layered
instances, it is also interesting to note that there is one instance
that can be solved by IMFP$^1$, but not by IMFP$^2$.

Additional analysis and experimentation revealed that in addition to
the difference $F-f$, the number of flow increases from $f$ to $F$
also seems to impact solution time, with more flow increases usually
resulting in longer solution times. The number of flow increases
from $f$ to $F$ depends on the arc capacities in an instance, and on
how these arc capacities interact on the different paths from source
to sink.

\subsection{Heuristics}\label{subsec:heuristics_computational}

For the difficult instances, we compare the performance of
\textit{Quickest-increment}, \textit{Quickest-to-ultimate}, and
\textit{Quickest-to-target} with $p = 2$, $r_0 = f, r_1 = \lfloor
(F-f) / 2 \rfloor$, and $r_2 = F$. The results are shown in Tables
\ref{tbl:alg_gen} and \ref{tbl:alg_layered}, where we report the
difference between the ultimate and initial flow values ($F-f$), for
each of the heuristics the cumulative flow (flow), the relative
difference to the best known cumulative flow ($\Delta =
(z^{\text{best}} - z^{\text{heur}})/z^{\text{best}}$), and the
solution time (time), and for IMFP$^2$ the value of the first
feasible solution (first), the relative difference to the best
solution ($\Delta = (z^{\text{best}} -
z^{\text{first}})/z^{\text{best}}$), the time to reach the first
feasible solution, the value of the best solution (best), and the
time to reach the best feasible solution (time). A time limit of
3,600 seconds was imposed for the solution of IMFP$^2$.
\begin{sidewaystable}[htp]
\caption{Comparison of heuristic performance for general graphs with
$n=35$, $d=0.3$, $p=0.7$, and $u_{\max}=10$.} \scriptsize
\begin{center}
\begin{tabular}{@{} rr rrr rrr rrr rrr rr @{}}
\toprule
 & & \multicolumn{3}{c}{\textit{Quickest-increment}} & \multicolumn{3}{c}{\textit{Quickest-to-ultimate}} & \multicolumn{3}{c}{\textit{Quickest-to-target}} & \multicolumn{5}{c}{IMFP$^2$} \\
\cmidrule(lr){3-5}\cmidrule(lr){6-8}\cmidrule(lr){9-11}\cmidrule(lr){12-16}
instance & $F-f$ & flow & $\Delta$ & time & flow & $\Delta$ & time & flow & $\Delta$ & time & first & $\Delta$ & time & best & time \\
    \midrule
    1     & 37    & 12,887 & 0.0002 & 1.20  & 12,884 & 0.0005 & 0.61  & 12,884 & 0.0005 & 0.70  & 12,868 & 0.0017 & 0.13  & 12,890 & 6.13 \\
    2     & 14    & 8,690 & 0.0000 & 0.37  & 8,690 & 0.0000 & 0.24  & 8,690 & 0.0000 & 0.34  & 8,690 & 0.0000 & 0.03  & 8,690 & 0.03 \\
    3     & 61    & 16,331 & 0.0021 & 2.23  & 16,347 & 0.0012 & 0.86  & 16,341 & 0.0015 & 0.91  & 16,304 & 0.0038 & 0.27  & 16,366 & 147.27 \\
    4     & 30    & 11,673 & 0.0005 & 1.07  & 11,679 & 0.0000 & 0.52  & 11,679 & 0.0000 & 0.61  & 11,676 & 0.0003 & 0.09  & 11,679 & 5.09 \\
    5     & 26    & 9,456 & 0.0002 & 0.74  & 9,456 & 0.0002 & 0.47  & 9,458 & 0.0000 & 0.60  & 9,453 & 0.0005 & 0.06  & 9,458 & 1.06 \\
    6     & 33    & 9,687 & 0.0003 & 1.06  & 9,689 & 0.0001 & 0.43  & 9,689 & 0.0001 & 0.55  & 9,658 & 0.0033 & 0.08  & 9,690 & 2.08 \\
    7     & 58    & 18,902 & 0.0002 & 1.48  & 18,904 & 0.0001 & 0.65  & 18,905 & 0.0000 & 0.77  & 18,808 & 0.0051 & 1.31  & 18,905 & 55.31 \\
    8     & 44    & 15,706 & 0.0000 & 1.25  & 15,705 & 0.0001 & 0.63  & 15,706 & 0.0000 & 0.70  & 15,701 & 0.0003 & 0.11  & 15,706 & 1.11 \\
    9     & 27    & 10,098 & 0.0001 & 1.01  & 10,097 & 0.0002 & 0.52  & 10,096 & 0.0003 & 0.52  & 10,078 & 0.0021 & 0.08  & 10,099 & 0.47 \\
    10    & 26    & 8,418 & 0.0001 & 0.85  & 8,418 & 0.0001 & 0.52  & 8,418 & 0.0001 & 0.56  & 8,418 & 0.0001 & 0.03  & 8,419 & 2.03 \\
    \midrule
    \textbf{} & \textbf{} & \textbf{} & \textbf{0.0004} & \textbf{1.126} & \textbf{} & \textbf{0.0002} & \textbf{0.545} & \textbf{} & \textbf{0.0003} & \textbf{0.626} & \textbf{} & \textbf{0.0017} & \textbf{0.219} & \textbf{} & \textbf{22.058} \\
    \bottomrule
\end{tabular}
\end{center}
\label{tbl:alg_gen}
\end{sidewaystable}%

\begin{sidewaystable}[htp]
\footnotesize \caption{Comparison of heuristic performance for
layered graphs with $\ell = 5$, $n = 10$, $d=0.3$, $p=0.7$, and
$u_{\max}=10$.} \scriptsize
\begin{center}
\begin{tabular}{@{} rr rrr rrr rrr rrr rr @{}}
\toprule
 & & \multicolumn{3}{c}{\textit{Quickest-increment}} & \multicolumn{3}{c}{\textit{Quickest-to-ultimate}} & \multicolumn{3}{c}{\textit{Quickest-to-target}} & \multicolumn{5}{c}{IMFP$^2$} \\
\cmidrule(lr){3-5}\cmidrule(lr){6-8}\cmidrule(lr){9-11}\cmidrule(lr){12-16}
instance & $F-f$ & flow & $\Delta$ & time & flow & $\Delta$ & time & flow & $\Delta$ & time & first & $\Delta$ & time & best & time \\
\midrule
    1     & 29    & 3,254 & 0.0157 & 1.22  & 3,295 & 0.0033 & 0.82  & 3,294 & 0.0036 & 0.84  & 3,249 & 0.0172 & 0.08  & 3,306 & 19.08 \\
    2     & 30    & 2,188 & 0.0271 & 1.21  & 2,233 & 0.0071 & 0.83  & 2,226 & 0.0102 & 0.80  & 2,190 & 0.0262 & 0.05  & 2,249 & 15.05 \\
    3     & 37    & 2,946 & 0.0081 & 0.89  & 2,956 & 0.0047 & 0.60  & 2,961 & 0.0030 & 0.64  & 2,879 & 0.0306 & 0.11  & 2,970 & 31.11 \\
    4     & 16    & 1,739 & 0.0063 & 0.49  & 1,743 & 0.0040 & 0.45  & 1,750 & 0.0000 & 0.50  & 1,731 & 0.0109 & 0.01  & 1,750 & 0.30 \\
    5     & 41    & 3,321 & 0.0160 & 1.26  & 3,352 & 0.0068 & 0.85  & 3,347 & 0.0083 & 0.93  & 3,323 & 0.0154 & 0.06  & 3,375 & 15.06 \\
    6     & 33    & 2,601 & 0.0069 & 1.14  & 2,612 & 0.0027 & 0.90  & 2,611 & 0.0031 & 0.94  & 2,525 & 0.0359 & 0.06  & 2,619 & 43.06 \\
    7     & 43    & 3,543 & 0.0250 & 1.41  & 3,604 & 0.0083 & 1.23  & 3,607 & 0.0074 & 1.06  & 3,505 & 0.0355 & 0.19  & 3,634 & 298.19 \\
    8     & 22    & 1,684 & 0.0024 & 0.57  & 1,685 & 0.0018 & 0.51  & 1,685 & 0.0018 & 0.56  & 1,661 & 0.0160 & 0.02  & 1,688 & 1.02 \\
    9     & 46    & 3,249 & 0.0107 & 1.25  & 3,206 & 0.0238 & 1.31  & 3,260 & 0.0073 & 1.10  & 3,194 & 0.0274 & 0.13  & 3,284 & 27.13 \\
    10    & 54    & 3,760 & 0.0079 & 1.34  & 3,736 & 0.0142 & 1.51  & 3,780 & 0.0026 & 1.23  & 3,652 & 0.0364 & 0.13  & 3,790 & 37.13 \\
    \textbf{} & \textbf{} & \textbf{} & \textbf{0.0126} & \textbf{1.078} & \textbf{} & \textbf{0.0077} & \textbf{0.901} & \textbf{} & \textbf{0.0047} & \textbf{0.86} & \textbf{} & \textbf{0.0252} & \textbf{0.084} & \textbf{} & \textbf{48.713} \\
\bottomrule
\end{tabular}
\end{center}
\label{tbl:alg_layered}
\end{sidewaystable}%

As expected, the three heuristics are able to produce high-quality
solutions quickly (with \textit{Quickest-increment} possibly being
slightly slower than the other two). For these instances, the hybrid
heuristic \textit{Quickest-to-target} seems to performs best (on
average within 0.03\% from the best known solution for instances on
general graphs and within 0.47\% from the best known solution for
instances on layered graphs). However, IMFP$^2$ also produces
high-quality solutions when given an hour of computing time, often
noticeably better than the heuristic solutions.

As the instances used for the above experiments are relatively
small, we performed a final experiment in which we use general and
layered graphs with 300 nodes to investigate whether the performance
of the heuristics (both in terms of quality and time) scales well.
Instances of this size are beyond what can be solved to optimality
in a reasonable amount of time with IMFP$^2$. Therefore, instead, we
solved the instances twice with Gurobi parameter
\texttt{SolutionLimit} set to one the first time and to two the
second time. When \texttt{SolutionLimit} is set to one, Gurobi uses
its embedded heuristics to generate a feasible solution and does not
even solve the linear programming relaxation. When
\texttt{SolutionLimit} is set to two, Gurobi solves at least one
linear programming relaxation.  The results can be found in Table
\ref{tab:heur-large} and Table \ref{tab:heur-large2}.

\begin{sidewaystable}[htp]
\footnotesize
  \caption{Comparison of heuristic performance on large instances;
  general graphs with $n = 300$, $d=0.3$, $p=0.7$, and $u_{\max}=10$.}
\scriptsize
\begin{center}
\begin{tabular}{rr  rr  rr  rr  rrrr}

    \toprule

 & & \multicolumn{2}{c}{\textit{Quickest-increment}} & \multicolumn{2}{c}{\textit{Quickest-to-ultimate}} & \multicolumn{2}{c}{\textit{Quickest-to-target}} & \multicolumn{4}{c}{IMFP$^2$} \\

\cmidrule(lr){3-4}\cmidrule(lr){5-6}\cmidrule(lr){7-8}\cmidrule(lr){9-12}

instance & $F-f$ & flow & time & flow & time & flow & time & first & time & second & time \\

    \midrule
    1     & 349   & \textbf{8,510,022} & 573.89 & 8,510,019 & 63.96 & 8,510,018 & 69.64 & 8,501,339 & 346.50 & \textbf{8,510,022} & 2,618.57 \\
    2     & 276   & \textbf{7,748,950} & 607.13 & 7,748,926 & 66.64 & 7,748,936 & 80.42 & 7,741,651 & 261.88 & \textbf{7,748,950} & 1,551.84 \\
    3     & 306   & \textbf{7,455,851} & 708.26 & 7,455,825 & 76.61 & 7,455,840 & 82.10 & 7,446,419 & 294.29 & \textbf{7,455,851} & 2,043.68 \\
    4     & 302   & 7,693,036 & 553.93 & 7,693,080 & 62.84 & 7,693,084 & 68.21 & 7,683,753 & 298.44 & \textbf{7,693,095} & 2,820.65 \\
    5     & 303   & \textbf{8,561,146} & 636.09 & 8,561,146 & 74.13 & 8,561,141 & 81.39 & 8,551,909 & 295.49 & \textbf{8,561,146} & 2,194.76 \\
    6     & 300   & \textbf{8,173,176} & 647.99 & 8,173,175 & 68.08 & 8,173,174 & 73.66 & 8,164,741 & 291.10 & 8,173,173 & 1,837.50 \\
    7     & 320   & \textbf{8,551,244} & 677.60 & 8,551,226 & 70.46 & 8,551,237 & 73.81 & 8,543,003 & 318.91 & \textbf{8,551,244} & 2,646.82 \\
    8     & 378   & \textbf{9,707,335} & 760.95 & 9,707,321 & 85.66 & 9,707,334 & 91.26 & 9,696,827 & 390.26 & \textbf{9,707,335} & 2,922.14 \\
    9     & 335   & \textbf{8,078,470} & 626.21 & 8,078,470 & 75.93 & 8,078,468 & 80.53 & 8,068,707 & 331.73 & \textbf{8,078,470} & 2,207.85 \\
    10    & 285   & \textbf{7,676,476} & 546.08 & 7,676,432 & 64.63 & 7,676,418 & 69.36 & 7,668,636 & 272.20 & \textbf{7,676,476} & 2,051.99 \\
    \midrule
    \textbf{} & \textbf{} & \textbf{} & \textbf{633.81} & \textbf{} & \textbf{70.894} & \textbf{} & \textbf{77.038} & \textbf{} & \textbf{310.08} & \textbf{} & \textbf{2,289.58} \\
    \bottomrule

\end{tabular}%
\end{center}
  \label{tab:heur-large}%
\end{sidewaystable}%

\begin{sidewaystable}[htp]
\footnotesize
  \caption{Comparison of heuristic performance on large instances;
  layered graphs with $\ell = 10$, $n = 30$, $d=0.3$, $p=0.7$, and $u_{\max}=10$.}
\scriptsize
\begin{center}
\begin{tabular}{rr  rr  rr  rr  rrrr}

    \toprule

 & & \multicolumn{2}{c}{\textit{Quickest-increment}} & \multicolumn{2}{c}{\textit{Quickest-to-ultimate}} & \multicolumn{2}{c}{\textit{Quickest-to-target}} & \multicolumn{4}{c}{IMFP$^2$} \\

\cmidrule(lr){3-4}\cmidrule(lr){5-6}\cmidrule(lr){7-8}\cmidrule(lr){9-12}

instance & $F-f$ & flow & time & flow & time & flow & time & first & time & second & time \\

    \midrule
    1     & 52    & 289,131 & 54.44 & 289,171 & 5.98  & \textbf{289,185} & 22.98 & 288,993 & 52.97 & 289,164 & 165.88 \\
    2     & 25    & 237,934 & 8.18  & 237,935 & 2.63  & \textbf{237,937} & 3.40  & 237,919 & 7.03  & 237,920 & 23.67 \\
    3     & 47    & 304,763 & 20.03 & 304,779 & 32.59 & \textbf{304,781} & 6.77  & 304,586 & 21.80 & 304,775 & 89.02 \\
    4     & 40    & \textbf{233,395} & 10.85 & 233,394 & 5.75  & \textbf{233,395} & 3.91  & 233,354 & 13.92 & 233,362 & 49.63 \\
    5     & 21    & 251,797 & 10.33 & 251,800 & 4.51  & \textbf{251,801} & 4.07  & 251,762 & 5.83  & 251,790 & 24.74 \\
    6     & 30    & 258,451 & 12.97 & 258,448 & 4.01  & 258,448 & 4.99  & 258,385 & 28.83 & \textbf{258,452} & 75.28 \\
    7     & 33    & \textbf{243,606} & 11.22 & 243,604 & 4.24  & 243,604 & 5.55  & 243,535 & 21.50 & 243,542 & 73.16 \\
    8     & 23    & \textbf{195,436} & 5.05  & 195,431 & 1.93  & \textbf{195,436} & 2.54  & 195,160 & 1.19  & 195,406 & 4.67 \\
    9     & 29    & \textbf{255,092} & 9.61  & 255,083 & 4.26  & 255,086 & 6.97  & 254,978 & 13.23 & 255,053 & 48.85 \\
    10    & 50    & 280,063 & 18.89 & 280,070 & 5.66  & \textbf{280,078} & 10.04 & 280,003 & 21.67 & 280,077 & 126.33 \\
    \midrule
    \textbf{} & \textbf{} & \textbf{} & \textbf{16.157} & \textbf{} & \textbf{7.156} & \textbf{} & \textbf{7.122} & \textbf{} & \textbf{18.797} & \textbf{} & \textbf{68.123} \\
    \bottomrule

\end{tabular}%
\end{center}
  \label{tab:heur-large2}%
\end{sidewaystable}%

A few interesting observations can be made. For general instances of
this size, \textit{Quickest-increment} performs best (it produces
the best solution for all but one instance), but it is noticeably
slower than \textit{Quickest-to-ultimate} and
\textit{Quickest-to-target}. The first solution found by IMFP$^2$ is
the worst for all instances, but the second solution found by
IMFP$^2$ matches the best in all but one instance. However, it takes
substantially longer to find that solution. For layered instances of
this size, the situation is not as clear cut, both
\textit{Quickest-increment} and \textit{Quickest-to-ultimate}
perform well and outperform IMFP$^2$ in all but one instance.
Overall, \textit{Quickest-to-target} produces high-quality solution
very efficiently for all instances and should be the method of
choice when time is important.

\section{Final remarks}\label{sec:final}

We have studied the incremental network design with maximum flows.
We have investigated the performance of mixed integer programming
formulations and we have analyzed the performance of natural
heuristics, both theoretically and empirically.

On the theoretic side, the complexity status of the incremental
maximum flow problem for instances with unit arc capacities remains
open, even when the network is restricted to bipartite graphs. On
the algorithmic side, we have identified classes of instances where
integer programming solvers struggle and where the natural
heuristics, although fast, do not necessarily provide high-quality
solutions. Thus, there is an opportunity to explore more
sophisticated heuristics, e.g., metaheuristics.

Finally, there are various other incremental network design problems
that are worth studying, e.g., the incremental multicommodity flow
problem.


\begin{thebibliography}{10}

\bibitem{baxter2014incremental}
M.~Baxter, T.~Elgindy, A.~Ernst, T.~Kalinowski, and M.~Savelsbergh.
\newblock Incremental network design with shortest paths.
\newblock {\em European Journal of Operational Research}, 238:675--684, 2014.

\bibitem{cavdaroglu2013integrating}
B.~Cavdaroglu, E.~Hammel, J.E. Mitchell, T.C. Sharkey, and W.A. Wallace.
\newblock Integrating restoration and scheduling decisions for disrupted
  interdependent infrastructure systems.
\newblock {\em Annals of Operations Research}, pages 1--16, 2013.

\bibitem{engel2013incremental}
K.~Engel, T.~Kalinowski, and M.W.P. Savelsbergh.
\newblock Incremental network design with minimum spanning trees.
\newblock arxiv:\href{http://arxiv.org/abs/1306.1926}{1306.1926}, 2013.

\bibitem{kim2008sequencing}
B.J. Kim, W.~Kim, and B.H. Song.
\newblock Sequencing and scheduling highway network expansion using a discrete
  network design model.
\newblock {\em The Annals of Regional Science}, 42(3):621--642, 2008.

\bibitem{lee2007restoration}
E.E. Lee, J.E. Mitchell, and W.A. Wallace.
\newblock Restoration of services in interdependent infrastructure systems: A
  network flows approach.
\newblock {\em IEEE Transactions on Systems, Man, and Cybernetics, Part C:
  Applications and Reviews}, 37(6):1303--1317, 2007.

\bibitem{lee2009network}
E.E. Lee, J.E. Mitchell, and W.A. Wallace.
\newblock Network flow approaches for analyzing and managing disruptions to
  interdependent infrastructure systems.
\newblock {\em Wiley handbook of science and technology for Homeland Security},
  2009.

\bibitem{mahmood2008design}
A.~Mahmood, M.~Aamir, and M.I. Anis.
\newblock Design and implementation of {AMR} smart grid system.
\newblock In {\em Electric Power Conference, 2008. EPEC 2008. IEEE Canada},
  pages 1--6. IEEE, 2008.

\bibitem{momoh2009smart}
J.A. Momoh.
\newblock Smart grid design for efficient and flexible power networks operation
  and control.
\newblock In {\em Power Systems Conference and Exposition}, pages 1--8. IEEE,
  2009.

\bibitem{nurre2012restoring}
S.G. Nurre, B.~Cavdaroglu, J.E. Mitchell, T.C. Sharkey, and W.A. Wallace.
\newblock Restoring infrastructure systems: An integrated network design and
  scheduling ({INDS}) problem.
\newblock {\em European Journal of Operational Research}, 223(3):794--806,
  2012.

\bibitem{nurre2013integrated}
S.G. Nurre and T.C. Sharkey.
\newblock Integrated network design and scheduling problems with parallel
  identical machines: Complexity results and dispatching rules.
\newblock {\em Networks}, 63(4):306--326, 2014.

\bibitem{UkPa2009}
S.~V. Ukkusuri and G.~Patil.
\newblock Multi-period transportation network design under demand uncertainty.
\newblock {\em Transportation Research Part B: Methodological}, 43(6):625--642,
  2009.

\end{thebibliography}

\newpage

\section*{Appendix}

\noindent \textbf{Complexity of Incremental Maximum Flow Problem}

\begin{theorem}\label{thm:hardness}
The incremental maximum flow problem is NP-hard even when restricted
to instances where every existing arc has capacity $1$ and every
potential arc has capacity $3$.
\end{theorem}
\begin{proof}
We use reduction from the problem \emph{Exact Cover by 3-sets}
(X3C). An instance of X3C is given by a collection $\mathcal S$ of
3-element subsets of the set $U=\{1,2,\ldots,3n\}$, and the problem
is to decide if there are $n$ sets $S_1,\ldots,S_n\in\mathcal S$
such that $U=S_1\cup\cdots\cup S_n$. Such an X3C instance can be
reduced to the following IMFP instance. The node set is
\[ N=\{s,t\}\cup\{v_S\ :\ S\in\mathcal S\}\cup\{w_1,\ldots,w_{3n}\}.\]
For every $S\in\mathcal S$, there is a potential arc with capacity
$3$ from $s$ to $v_S$, and for every $i\in U$ there is an existing
arc of capacity $1$ from $w_i$ to $t$. Moreover, for every
$S=\{i,j,k\}\in\mathcal S$ there are three existing arcs
$(v_S,w_i)$, $(v_S,w_j)$ and $(v_S,w_k)$ with unit capacities.
Clearly the flow in time period $k$ is at most
$3\cdot\min\{n,k-1\}$, hence the objective value is bounded by
\[3\cdot[0+1+\cdots+(n-1)]+3n(T-n)\]
and this upper bound can be achieved if and only if the underlying
X3C instance is a YES-instance.

\end{proof}

\noindent \textbf{Approximation for the incremental maximum matching
problem}

\vspace{12pt} \noindent In order to derive a performance guarantee
for \textit{Quickest-to-ultimate} on instances of the incremental
maximum matching problem, we strengthen
Lemma~\ref{lem:z_i_by_lambda_resp_mu}. We need the following
auxiliary result.
\begin{lemma}\label{lem:averaging_is_best}
For real numbers $\alpha_1\geqslant\alpha_2\geqslant\cdots\geqslant\alpha_n$ and $0\leqslant\beta_1\leqslant\beta_2\leqslant\cdots\leqslant\beta_n$ with $\sum_{i=1}^n\beta_i=B$, we have
\[\sum_{i=1}^n\alpha_i\beta_i\leqslant\frac{B}{n}\sum_{i=1}^n\alpha_i.\]
\end{lemma}
\begin{proof}
By duality,
\begin{multline*}
\max\left\{\sum_{i=1}^n\alpha_ix_i\ :\ x_i\leqslant x_{i+1}\text{ for }1\leqslant i\leqslant n-1,\ \sum_{i=1}^nx_i=B,\ x_i\geqslant 0\text{ for }1\leqslant i\leqslant n\right\} \\
=\min\left\{Bz\ :\ y_i-y_{i-1}+z\geqslant\alpha_i\text{ for }1\leqslant i\leqslant n,\ y_0=y_n=0,\ y_i\geqslant 0\text{ for }1\leqslant i\leqslant n-1 \right\},
\end{multline*}
and a feasible solution for the minimization problem on the RHS is
given by $z=1/n\sum_{i=1}^n\alpha_i$ and
\[
y_{i}=(n-i)z-\sum_{j=i+1}^n\alpha_j\qquad\text{for }1\leqslant i\leqslant n-1.\qedhere
\]
\end{proof}
\begin{lemma}\label{lem:improved_bound_for_z1}
For the incremental maximum matching problem,
\[z_1\geqslant
\begin{cases}
TF-\frac12c_r(r+1) & \text{if }f\geqslant r, \\
TF-\frac14\left[c_rF+r(r-f)+2c_r\right] & \text{if }f<r.
\end{cases}\]
\end{lemma}
\begin{proof}
Recall that $\sum_{i=0}^{r-1}\lambda_i=c_r$.
Lemmas~\ref{lem:z_i_by_lambda_resp_mu}
and~\ref{lem:averaging_is_best} yield
\begin{multline*}
z_1\geqslant TF-\sum_{i=0}^r\lambda_i(r-i)\geqslant TF-\frac{c_r}{r-1}\sum_{i=0}^r(r-i)=TF-\frac{c_r}{r}\left(r^2-\frac12r(r-1)\right)\\
=TF-\frac12c_r(r+1).
\end{multline*}
For the case $r>f$ we define $\rho=\max\{0,\lceil(r-f)/2\rceil\}$.
Note that $f+\rho-1<f+(r-f)/2=F/2$. This implies that in the first
$\rho$ time periods, there is always a potential arc which is used
in the ultimate maximum matching, and is not adjacent to any arc in
the current maximum matching. Hence $\lambda_i=0$ for $0\leqslant
i\leqslant\rho-1$. From Lemma~\ref{lem:z_i_by_lambda_resp_mu} we
obtain
\[
z_1\geqslant TF-\sum_{i=0}^{\rho-1}(r-i)-\sum_{i=\rho+1}^{r-1}\lambda_i(r-i).
\]
Using Lemma~\ref{lem:averaging_is_best} and
$\displaystyle\sum_{i=\rho}^{r-1}\lambda_i=c_r-\rho$, this implies
\begin{multline*}
z_1\geqslant TF-r\rho+\frac{\rho(\rho-1)}{2}-\frac{c_r-\rho}{r-\rho}\sum_{i=\rho+1}^{r-1}(r-i)\\
=TF-r\rho+\frac{\rho(\rho-1)}{2}-\frac{c_r-\rho}{r-\rho}\left[r(r-\rho)-\frac12(r-t)(r+t-1)\right]\\
=TF-r\rho+\frac{\rho(\rho-1)}{2}-\frac{(c_r-\rho)(r-\rho+1)}{2}.
\end{multline*}
If $\rho=(r-f)/2$ then
\begin{multline*}
z_1\geqslant TF-\frac12r(r-f)+\frac{(r-f)(r-f-2)}{8}-\frac{(c_r-(r-f)/2)(r-(r-f)/2+1)}{2}\\
=TF-\frac14\left[c_rF+r(r-f)+2c_r)\right],
\end{multline*}
and if $\rho=(r-f)/2$ then
\begin{multline*}
z_1\geqslant TF-\frac12r(r-f+1)+\frac{(r-f+1)(r-f-1)}{8}-\frac{(c_r-(r-f+1)/2)(r-(r-f+1)/2+1)}{2}\\
=TF-\frac14\left[c_rF+r(r-f)+c_r+r)\right].
\end{multline*}
Now the claim follows from $c_r\geqslant r$.
\end{proof}
\begin{proposition}\label{prop:upper_bound}
\textit{Quickest-to-ultimate} is a $4/3$-approximation algorithm for
the incremental maximum matching problem.
\end{proposition}
\begin{proof}
We need to show that $z^*-\frac43z_1\leqslant 0$, and we distinguish
two cases.
\begin{description}
\item[Case 1.] $r\leqslant f$, i.e. $\rho=0$. Using Lemmas~\ref{lem:bound_opt_by_c} and~\ref{lem:improved_bound_for_z1}, we obtain
\[z^*-\frac43z_1\leqslant TF-\sum_{j=1}^rc_j\leqslant TF-\frac{r(r-1)}{2}-c_r-\frac43\left[TF-\frac{c_r(r+1)}{2}\right]\leqslant\frac{2c_rr}{3}-\frac{TF}{3}.\]
The required inequality follows from $T\geqslant c_r$ and $F=f+r\geqslant 2r$.
\item[Case 2.] $r>f$. Lemma~\ref{lem:improved_bound_for_z1} implies
  \begin{multline*}
z^*-\frac43z_1\leqslant TF-\sum_{j=1}^rc_j\\
\leqslant TF-\frac{r(r-1)}{2}-c_r-\frac43\left[TF-\frac14\left[c_rF+r(r-f)+2c_r\right]\right]\\
= -\frac{TF}{3}-\frac{r(r-1)}{2}-c_r+\frac{c_rF}{3}+\frac{r(r-f)}{3}+\frac{2c_r}{3}.
 \end{multline*}
Using $T\geqslant c_r+1$ and then substituting $F=f+r$, we obtain
  \begin{multline*}
z^*-\frac43z_1\leqslant -\frac{c_rF}{3}-\frac{F}{3}-\frac{r^2}{2}+\frac{r}{2}-c_r+\frac{c_rF}{3}+\frac{r^2}{3}-\frac{rf}{3}+\frac{2c_r}{3}\\
\leqslant -\frac{f}{3}-\frac{r}{3}+\frac{r}{2}-c_r-\frac{rf}{3}-\frac{c_r}{3}\leqslant \frac{r}{6}-\frac{c_r}{3}.
 \end{multline*}
Now the claim follows from $c_r\geqslant r$.\qedhere
\end{description}
\end{proof}

\end{document}